\documentclass[final,5p,times,twocolumn, lefttitle]{elsarticle}

\usepackage{amssymb}

\usepackage[T1]{fontenc}
\usepackage{amsfonts}
\usepackage{amsmath}
\usepackage{amsthm}
\usepackage{bbm}
\usepackage{xcolor}
\usepackage{mathrsfs}
\usepackage{hyperref}
\hypersetup{colorlinks, linkcolor=red, citecolor=blue, urlcolor=blue}
\usepackage{algorithm}
\usepackage{algorithmic}
\usepackage{lineno}
\modulolinenumbers[5]

\newcommand{\CC}{\mathbb{C}}
\newcommand{\EE}{\mathbb{E}}
\newcommand{\PP}{\mathbb{P}}
\newcommand{\RR}{\mathbb{R}}

\newcommand{\cA}{\mathcal{A}}
\newcommand{\cB}{\mathcal{B}}
\newcommand{\cE}{\mathcal{E}}
\newcommand{\cF}{\mathcal{F}}
\newcommand{\cI}{\mathcal{I}}
\newcommand{\cM}{\mathcal{M}}
\newcommand{\cN}{\mathcal{N}}
\newcommand{\cP}{\mathcal{P}}
\newcommand{\cS}{\mathcal{S}}
\newcommand{\cV}{\mathcal{V}}
\newcommand{\cW}{\mathcal{W}}

\newcommand{\sI}{\mathscr{I}}

\DeclareMathOperator*{\argmin}{arg\,min}
\DeclareMathOperator*{\argmax}{arg\,max}

\let\eps\varepsilon
\let\to\longrightarrow

\newcommand{\conv}[2][\ ]{\overset{#1}{\underset{#2}{\to}}}

\newcommand{\aseq}[1]{\underset{#1}{=}}
\newcommand{\equi}[1]{\underset{#1}{\sim}}
\newcommand{\mbf}[1]{\mathbf{#1}}

\theoremstyle{definition}

\newtheorem{defi}{Definition}

\theoremstyle{plain}

\newtheorem{lem}{Lemma}

\usepackage{colortbl}

\journal{Elsevier}

\begin{document}
\begin{frontmatter}

\title{
Robust a posteriori estimation of 
probit-lognormal seismic fragility curves via 
sequential
design of experiments and constrained reference prior
}

\author[label1,label2]{Antoine Van Biesbroeck\corref{cor1}}
\ead{antoine.van_biesbroeck@ens-paris-saclay.fr}\ead[url]{https://vbkantoine.github.io}

\author[label3]{Clément Gauchy}
\author[label2]{Cyril Feau}
\author[label1]{Josselin Garnier}

\affiliation[label1]{organization={CMAP, CNRS, École polytechnique, Institut Polytechnique de Paris},%
            city={91120 Palaiseau},
            country={France}}

\affiliation[label2]{organization={Université Paris-Saclay, CEA, Service d'Études M\'ecaniques et Thermiques},%
            city={91191 Gif-sur-Yvette},
            country={France}}

\affiliation[label3]{organization={Université Paris-Saclay, CEA, Service de Génie Logiciel pour la Simulation},%
            city={91191 Gif-sur-Yvette},
            country={France}}

\cortext[cor1]{Corresponding author}

\begin{abstract}
{
A seismic fragility curve expresses the probability of failure of a structure conditional to an intensity measure (IM) derived from seismic signals. When only limited data is available, the practitioner often refers to the probit-lognormal model coupled with maximum likelihood estimation (MLE) to obtain estimates of these curves. This means that only a binary indicator of the state (BIS) of the structure is known, namely a failure or non-failure state indicator, when it is subjected to a seismic signal with an intensity measure IM. 
In this context, the objective of this work is to propose a method for optimally estimating such curves by obtaining the most precise estimate possible with the minimum of data. The novelty of our work is twofold. First, we present and show how to mitigate the likelihood degeneracy problem which is ubiquitous with small data sets and hampers frequentist approaches such as MLE. Second, we propose a novel strategy for sequential design of experiments (DoE) that selects seismic signals from a large database of synthetic or real signals via their IM values, to be applied to structures to evaluate the corresponding BISs. This strategy relies on a criterion based on information theory in a Bayesian framework. It therefore aims to sequentially designate the IM value such that the pair (IM, BIS) has on average, with respect to the BIS of the structure, the greatest impact on the posterior distribution of the fragility curve. The methodology is applied to a case study from the nuclear industry. The results demonstrate its ability to efficiently and robustly estimate the fragility curve, and to avoid degeneracy even with a limited amount of data, {i.e., less than 100}. Furthermore, we demonstrate that the estimates quickly reach the model bias induced by the probit-lognormal modeling. Eventually, two criteria are suggested to help the user stop the DoE algorithm.}
\end{abstract}


\end{frontmatter}


{
\section{Introduction}

Fragility curves are key ingredients of the seismic probabilistic risk assessment (SPRA) framework, introduced in the 1980s for seismic risk assessment studies carried out on nuclear facilities \cite{Kennedy1980,KENNEDY198447, ELLINGWOOD1988, PARK1998, KENNEDY1999, Cornell2004}. They are also part of the performance-based earthquake engineering framework (PBEE) \citep{GHOBARAH2001878, Cornell2004, Luco2007, Noh2014} since they represent a convenient means to measure the seismic performance of structures under seismic excitation, in probabilistic terms. In the two cases, the seismic hazard ---driven by the magnitude (M), the source-site distance (R), and other earthquake parameters--- is reduced to a scalar value derived from the seismic ground motion ---the intensity measure (IM)--- under the so-called ``sufficiency assumption'' \cite{Cornell2004,Luco2007}. In practice, fragility curves therefore express the probability of failure of a mechanical structure as a function of an IM value of interest such as peak ground acceleration (PGA) or pseudo-spectral acceleration (PSA). It should be noted that the sufficiency assumption was introduced to reduce estimation costs since it assumes that the fragility curve of a given structure is identical regardless of the seismic scenario. As shown in \cite{Radu2018,Grigoriu2021}, different seismic scenarios can however lead to identical distributions of some IMs, although the underlying seismic signals have significantly different frequency contents. As a result, the sufficiency assumption is not necessarily met, especially for non-linear, multimodal structures, etc. It is however possible to focus on the definition of the resulting fragility curve, without taking the assumption for granted, which is the case in this work.

If the available data set contains more information than just a binary indicator of the state (BIS) of the structure ---i.e., a failure or non-failure state indicator--- when subjected to seismic loading, machine learning-based techniques can be leveraged to estimate fragility curves. For example, let us quote: linear regression or generalized linear regression \cite{Lallemant2015, Zentner2017, Mai2017, VanBiesbroeckESAIMProcS}, classification-based techniques \cite{BERNIER2019, KIANI2019, Sainct2020}, kriging \cite{Gidaris15, GENTILE2020101980,Kyprioti2021, Gauchy2022IJUQ,Yi2024}, polynomial chaos expansion \cite{Mai16}, stochastic polynomial chaos expansions \cite{ZHU2023}, and artificial neural networks \cite{WANG2018232, MITROPOULOU2011, WANGZ2018}. Whenever data is obtained through numerical simulations, some of these methods can be coupled with adaptive techniques to reduce the number of calculations required \cite{WANG2018232, Sainct2020, Gidaris15}.

If the available information is limited to the BIS of the structure under seismic excitation, the use of a parametric model of the fragility curve is mandatory. Historically, in the framework of SPRA and PBEE, the probit-lognormal model was chosen and it remains prevalent to this day due to its proven capability to handle limited data from various sources: (i) expert assessments supported by test data \cite{Kennedy1980, KENNEDY198447, PARK1998, Zentner2017}, (ii) experimental data \cite{PARK1998}, (iii) empirical data from past earthquakes \cite{Shinozuka2000, Straub2008, Lallemant2015, Buratti2017, Laguerre2024}, and (iv) analytical results obtained from various numerical models using artificial or natural seismic excitations \cite{ELLINGWOOD2001251, KIM2004, ZENTNER20101614, KOUTSOURELAKIS2010, Zentner2017, Mai2017, TREVLOPOULOS2019, WangF2020,MANDAL201611, WANGZ2018, WANG2018232, ZHAO2020103,
Katayama2021,Gauchy2021,KHANSEFID2023,LEE2023,VanBiesbroeck2023, VanBiesbroeckUNCECOMP2023, VanBiesbroeckESAIMProcS, MAHANTA2025}.

The present work is part of the context just mentioned in which only partial knowledge of the structural response under seismic excitation is available. Several strategies can then be implemented to estimate the two parameters of the probit-lognormal model, namely the median $\alpha$ and the log of the standard deviation $\beta$. Some of them are compared and their strengths and weaknesses are highlighted in \cite{Lallemant2015}. Thus, when the data is binary, \citeauthor{Lallemant2015} \cite{Lallemant2015} and \citeauthor{Mai2017} \cite{Mai2017} both recommend using maximum likelihood estimation (MLE). Additionally, when the data samples are independent of each other, the bootstrap technique can be used to obtain confidence intervals related to the size of the samples considered \cite{Shinozuka2000, ZENTNER20101614, WangF2020}. Recently, \citeauthor{VanBiesbroeck2023} \cite{VanBiesbroeck2023} demonstrated however that MLE can lead to unrealistic or degenerate fragility curves, such as unit-step functions, when data is sparse, because the likelihood itself can be degenerate. This results from the composition of the data set. For example, this occurs when no failures are observed, etc. Bayesian approaches can help to overcome these likelihood degeneracy issues. 
They are becoming increasingly popular in seismic fragility analysis since they are seen as a way to integrate prior knowledge (expert judgment, etc.) with newly acquired data \cite{Gardoni2002, Straub2008, WANG2018232, Katayama2021, LEE2023, KOUTSOURELAKIS2010, damblin2014, TADINADA201749, KWAG20181, Jeon2019, TABANDEH2020, VanBiesbroeck2023, VanBiesbroeckESAIMProcS, MAHANTA2025}. However, for posterior estimates to be meaningful (e.g., credibility intervals), the prior must be carefully defined or chosen. In the SPRA framework, Bayesian inference is often used to update existing probit-lognormal fragility curves previously obtained through various approaches assuming, for example, lognormal distribution for $\alpha$ \cite{WANG2018232}, 
independent distributions for the prior values of $\alpha$ and $\beta$ such as uniform distributions \cite{KOUTSOURELAKIS2010}, a normal distribution for $\ln(\alpha)$ and the improper distribution $1/\beta$ for $\beta$ \cite{Straub2008}, a normal distribution for $\ln(\alpha)$ and a uniform distribution for $\beta$  \cite{MAHANTA2025}, etc. In \cite{VanBiesbroeck2023, VanBiesbroeckUNCECOMP2023, VanBiesbroeckESAIMProcS}, the authors favor the framework of the theory of reference priors to define the prior on $\alpha$ and $\beta$. According to the metrics introduced in \cite{Kass1996}, this theory allows to determine a prior which is qualified as ``objective''. Specifically, the criterion introduced by \citeauthor{Bernardo1979} \cite{Bernardo1979} aims at identifying the prior that maximizes the capacity to ``learn'' from observations. The prior is therefore selected to maximize a mutual information indicator, which expresses the information provided by the data to the posterior, relatively to the prior. In \cite{VanBiesbroeck2023, VanBiesbroeckUNCECOMP2023} authors showed that the resulting prior, namely the Jeffreys prior, depends only on the distribution of the IM of interest. This prior is therefore suitable for all equipments in an industrial installation because they are subject to the same seismic hazard, i.e., to the same seismic signals corresponding to the seismic scenario of interest and which must be used for the calculations or tests of said installation. With limited data sets, authors showed that this approach outperforms the classical approaches of the literature both in terms of regularization (fewer degenerate estimations) and stability (absence of outliers when sampling the posterior distribution of the parameters). Finally, let us note that the Bayesian framework is also relevant for fitting numerical models (e.g., mathematical expressions based on engineering assessments or physics-based models) to experimental data in order to estimate fragility curves \cite{Gardoni2002, TABANDEH2020} or meta-models such as logistic regressions \cite{KOUTSOURELAKIS2010, Jeon2019}.

{This work is an extension of the work presented by the authors in the references \cite{VanBiesbroeck2023, VanBiesbroeckUNCECOMP2023}. To go further, we propose a design of experiments (DoE) strategy to make the most of the probit-lognormal model by obtaining the most accurate estimate possible with the minimum amount of data. We therefore propose a criterion inherited from the reference prior theory, which aims to sequentially designate the value of the IM such that the pair (IM, BIS) has on average, with respect to the BIS of the structure, the greatest impact on the posterior distribution of the fragility curve. In practice, this allows (i) the sequential selection of seismic signals via their IM values, from a large database of synthetic signals generated to match seismic scenario of interest, in order to (ii) evaluate the associated structural responses to derive the BISs. For brevity, we will refer to this process as “data selection” in the following.}

{In Section~\ref{sec:contrib}, the positioning, contribution and methodology of this paper are first presented.} The Bayesian framework is the subject of Section~\ref{sec:model}. In particular, the definitions of ``likelihood degeneracy'' and ``robust a posteriori estimation'' are given there. Section~\ref{sec:PEmethod} is devoted to the presentation of our DoE strategy. Then, in Section~\ref{sec:application} a comprehensive study of the performance of this strategy is presented considering a case study from the nuclear industry. The conclusion is finally proposed in Section~\ref{sec:conclusion}. We specify that the results presented in the paper are also supported by the study of a ``toy case'' which is presented in \ref{app:toycases}.

{
\section{Positioning, contribution and methodology \label{sec:contrib}}

\subsection{Positioning}

This paper focuses on the estimation of seismic fragility curves using the lognormal model in the case of limited (typically less than 100) and binary data, i.e., they only reflect the state of the structure (failure or non-failure) when subjected to seismic excitation. Therefore, it mainly addresses the issues of equipment for which only binary seismic test results are available (for example, electrical relay qualification tests, etc.). Nevertheless, this also applies to simulation-based approaches whose results are reduced, by simple post-processing, to binary data. Numerical simulations provide much richer information than the simple binary state of a structure. Thus, they allow the use of more or less sophisticated statistical models, such as those mentioned in the introduction, to capture the relationship between the engineering demand parameter (EDP) and the IM, which can be complex, particularly when the behavior of the structure becomes non linear under the effect of earthquakes (which is the case of the industrial structure considered in this work). In such situations, when available data is limited, choosing the most appropriate statistical model is not a trivial task. In \cite{Kyprioti2021,Gauchy2022IJUQ,Yi2024}, the authors highlight that a major challenge for the application of surrogate models as part of earthquake engineering, is related to the random variability associated with seismic hazard which exhibits heteroscedastic behavior. This may result in using a large amount of data to obtain reliable estimates. In \cite{Mai2017} and \cite{VanBiesbroeckESAIMProcS}, the authors showed that the simplest model ---involving less data to adjust its parameters--- assuming a linear correlation between the logarithm of the EDP and that of the IM should be avoided due to the irreducible model bias that it can introduce and which can greatly affect the estimation of fragility curves. In their setting, \cite{Mai2017} showed that the probit-lognormal model gives better estimates than those based on linear regression.  Although the use of a parametric model for the estimation of seismic fragility curves is subject to debate, its abundant use in the literature shows that the lognormal model appears to the practitioner as a model that is both pragmatic and relevant with a reduced number of data. Let us note that recent studies on the impact of IMs on fragility curves \cite{Sainct2020, CIANO2020, CIANO2022} suggest that the choice of an appropriate IM makes it possible to reduce the potential biases between reference fragility curves and their probit-lognormal estimations. Finally, let us add that regardless of the statistical model or strategy chosen to estimate seismic fragility curves, one way to minimize the amount of data to be used is to implement adaptive strategies. This is particularly true for structures with a high failure threshold relative to the seismic scenario of interest, even with the use of more sophisticated models than the probit-lognormal model \cite{ZHU2023}. 

\subsection{Contribution}

For all the reasons just mentioned, we seek in this paper to take full advantage of the probit-lognormal model by proposing a DoE strategy within a Bayesian framework. We take as support the reference prior theory in order to define an objective prior adapted to the probit-lognormal model. For that, we relied on the work  \cite{VanBiesbroeck2023} in which the authors derived the mathematical expression of the Jeffreys prior of the probit-lognormal model. Then, in order to facilitate the implementation and improve the numerical efficiency of the DoE algorithm, we propose a simplified analytical expression of the Jeffreys prior which takes up its main features. This expression also builds on recent developments presented in \cite{VanBiesbroeck2024constraints} that aim to slightly constrain the prior to ensure that it is proper. In our framework, this addresses the problems related to likelihood degeneracy, which are ubiquitous in small data sets and involve an improper posterior \cite{VanBiesbroeck2023}. Since such posterior cannot be normalized as a probability, it prevents the use of MCMC algorithms. This is obviously a crucial point for any posterior estimation. These developments are presented in detail in this article and constitute, besides the DoE strategy itself, another major contribution of this work. Eventually, two criteria are suggested to help the user stop the DoE algorithm.

Note that a new concept of robust fragility has emerged in the literature over the last two decades. It derives from the concept of robust reliability defined in \cite{PAPADIMITRIOU2001} for the assessment of the reliability of structures. In this framework, ``robustness'' means that modeling uncertainties are explicitly taken into account, so that the calculated reliability is not sensitive to them. In the context of this paper, the robust fragility curve refers to the average fragility curve obtained by integration over the possible values of the parameters of the lognormal model, whose density is defined by the posterior distribution, and the associated credibility interval is quantified by its variance \cite{Jalayer2015}. Knowing that with a limited amount of data, the probability of having a degenerate likelihood can be significant and that this can compromise any a posteriori estimation (the posterior can be improper, it depends on the prior) we propose to define a robust posterior estimation as an estimation based on a proper posterior. The definition that we propose presents the conditions that must be satisfied for the posterior to be proper. It therefore defines the conditions necessary for the estimation of a robust fragility curve.

\subsection{Methodology : principle and validation}

Our methodology aims to optimally select seismic signals from a wide range of signals corresponding to a seismic scenario of interest. Since real signal databases are generally not rich enough to allow such a selection, we use in the first place a seismic signal generator (SSG) to enrich them. Among the set of SSGs available in the literature, we used the parameterized stochastic model of modulated and filtered white-noise process defined in \cite{Rezaeian2010}. This generator has been used several times since 2010 \cite{Mai16,Mai2017,Vassiliou2017,Sainct2020,Gauchy2022IJUQ,
VanBiesbroeck2023, VanBiesbroeckESAIMProcS,ZHU2023}. In our experience, it efficiently addresses both temporal and spectral nonstationarities of seismic signals. From an engineering point of view, the current IMs of the generated artificial signals have statistical characteristics close to those of the real signals used for enrichment. For the present work, we rely on the $10^5$ synthetic signals generated by \citeauthor{Sainct2020} \cite{Sainct2020}. They were obtained by calibrating the SSG from 97 real accelerograms selected in the European Strong Motion Database for a magnitude $M$ such that $5.5 \leq M \leq 6.5$, and a source-to-site distance $R < 20$~km~\cite{ESMD}.

It should be note that the DoE algorithm proposed in this work is independent of the generator and the seismic scenario. The study of the influence of the SSG is out of the scope of this work, which aims to propose a general methodological framework for the selection of seismic signals, for an optimal estimation of fragility curves.

For illustrative and performance evaluation purposes, the proposed methodology is applied to a case study taken from the nuclear industry: a piping system. We use here a numerical model based on beam elements, validated by comparison with tests performed on the Azalee shaking table at CEA/Saclay \cite{TOUBOUL1999}. This model that could today be described as a "low-fidelity model", offers a good trade-off between accuracy and computation time. It allowed us to perform $8\cdot 10^4$ calculations, by randomly selecting $8\cdot 10^4$ excitation signals from among the $10^5$ artificial signals. Such a quantity of calculations would not have been reasonably achievable with a more sophisticated numerical model. In our setting, approximately 14\% of these calculations involve nonlinear behavior (material nonlinearity) of the piping system. This represents approximately 90 days of sequential calculations on a simple laptop. This numerical effort was carried out for comparative analysis purposes, but cannot be implemented in practice on the scale of an industrial installation. That is, in fact, the objective of the proposed methodology, namely to enable a robust estimation of fragility curves involving a minimum of data (i.e., less than 100).

Thus, the $8\cdot 10^4$ calculations allow estimation of a reference fragility curve, against which the estimates of our method can be compared. This large number of results also allows the data selection methodology to be replicated many times, considering subsets of smaller size than the total validation sample. Each replication consists of randomly choosing two initial pairs (IM, BIS) and adding new pairs by sequentially applying the DoE algorithm up to 250 selections. Compared to a standard sampling method, these replications allow for statistical validation of the DoE method, based on performance metrics, such as the mean squared error, etc.}
}

\section{Bayesian estimation of seismic fragility curves}\label{sec:model}

    \subsection{{Statistical} Model and likelihood}

In this work, we address case studies within which the %
observation of the equipment's response to the ground-motion %
is limited to a binary outcome: failure {or non-failure.}
We denote by $Z$ the random variable that equals $1$ if the equipment fails, and $0$ otherwise. The value of $Z$ is determined by uncertainties that comes from the mechanical model and from the seismic scenario. The latter is summarized, {in a non-exhaustive manner}, into a chosen IM that is stochastic. The associated random variable is denoted by $A$, it takes its values in a set $\cA\subset(0,\infty)$.

In those settings, it is common to consider a probit-lognormal statistical modeling of the fragility curve:
    \begin{equation}\label{eq:probit-model}
        P_f(a) = \PP(Z=1|A=a,\,\theta) = \Phi\left(\frac{\log a-\log\alpha}{\beta}\right),
    \end{equation}
    
{
where $\Phi$ is the cumulative distribution function of a standard normal distribution and $\theta:=(\alpha,\beta)$ designates parameters that need to be estimated. $\alpha\in(0,\infty)$ is the median of the fragility curve and $\beta\in(0,\infty)$ is the log standard deviation, so that $\theta\in\Theta=(0,\infty)^2$.
}

Given $k\geq1$ and an observed sample of independent and identically distributed realizations of the random couple $(Z,A)$: $(\mbf z^k, \mbf a^k)$, $\mbf z^k=(z_i)_{i=1}^k$, $\mbf a^k=(a_i)_{i=1}^k$, the statistical model described above gives the following likelihood, denoted by $\ell_k(\mbf z^k,\mbf a^k|\theta)$:
    \begin{align}\label{eq:likelihood}
        &\ell_k(\mbf z^k,\mbf a^k|\theta)\\ &= \prod_{i=1}^k\Phi\left(\frac{\log a_i-\log\alpha}{\beta}\right)^{z_i}\left(1-\Phi\left(\frac{\log a_i-\log\alpha}{\beta}\right)\right)^{1-z_i}p_A(a_i),\nonumber
    \end{align}
with $p_A$ denoting the probability density function of the distribution of the IM.
We let the reader notice that the distribution $p_A$ may influence the likelihood. 
This point is further discussed in Section~\ref{sec:prior}, it emphasizes the importance to have access to a realistic distribution of the IM. 
The practical distribution considered for this work is detailed in Section~\ref{sec:application}.
{It is common in the literature (e.g. \cite{VanBiesbroeck2023}) to define the likelihood as the probability density function of $\mbf z^k$ conditionally to $(\mbf a^k,\theta)$. That density, denoted by $\ell_k(\mbf z^k|\mbf a^k,\theta)$, equals the product in Eq.~(\ref{eq:likelihood}) without the the terms $(p_A(a_i))_{i=1}^k$. However, even with this convention, $p_A$ is eventually involved in the same way.} 

\subsection{Degeneracy of the likelihood}\label{sec:likelihooddegen}

{In \cite{VanBiesbroeck2023}, the authors studied the asymptotic behavior of $\ell_k(\mbf z^k|\mbf a^k,\theta)$ w.r.t. $\theta$.} Their result lead to the understanding of a critical phenomenon which we call the ``degeneracy'' of the likelihood.

\begin{defi}[Likelihood degeneracy]\label{def:degeneracy}
    {If the observed samples $(\mbf z^k,\mbf a^k)$ belong to one of the following three types:}
    \begin{itemize}
        \item {type 1 : no failure is observed: $z_i=0$ for any $i$;}
        \item {type 2 : only failures are observed: $z_i=1$ for any $i$;}
        \item {type 3 : the failures and non-failures are partitioned into two disjoint subsets when classified according to their IM values:}
        there exists $a\in\cA$ such that for any $i,j$, $a_i<a<a_j\Longleftrightarrow z_i\ne z_j$ (see the illustration in Fig.~\ref{fig:degenerate-frag});
    \end{itemize}
    then the likelihood is called \emph{degenerate}.
\end{defi}

An example of a degenerate likelihood is presented in Fig.~\ref{fig:degenerate-frag}.
{In practice, this degeneracy constitutes a major obstacle to the use of a method such as MLE coupled with the bootstrap technique. Within the Bayesian framework, it can compromise posterior sampling by providing improper posteriors. Whether with frequentist or Bayesian methods, this degeneracy can be at the origin of critical estimates of fragility curves such as unit step functions. This phenomenon is quite frequent when a limited number of data samples are observed. {In our context, where we consider a large set of seismic signals generated beforehand to match a seismic scenario of interest,} it is more pronounced when (i) failures rarely occur with respect to the distribution of seismic signals, and (ii) the response of the structure is more correlated to the considered IM because, in this case, the value of $\beta$ of the associated fragility curve tends towards 0 \cite{VanBiesbroeckUNCECOMP2023}. In Section~\ref{sec:application}, we discuss the occurrence of degeneracy in the illustrative case study.}

\begin{figure*}
    \centering
    \includegraphics[width=4.44cm]{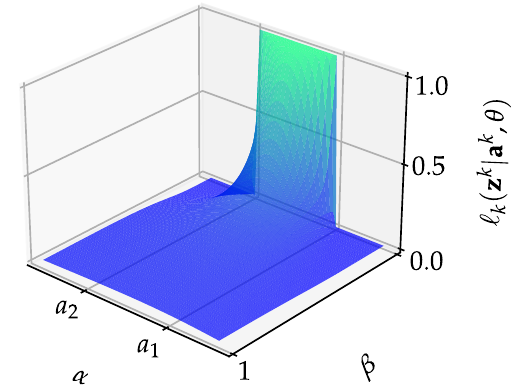}\qquad
    \includegraphics[width=4.44cm]{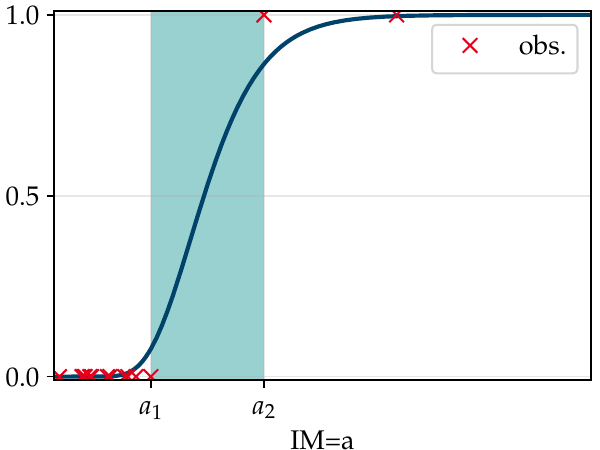}\\[0.7em]
    \begin{minipage}{0.94\textwidth}\footnotesize \itshape Graphs of (left) the likelihood given the degenerate data sample as a function of the tuple $(\alpha,\beta)$ and (right) the fragility curve (blue curve) according to which the points (red crosses) are sampled : each red cross is a tuple $(z,a)$, where its X-axis equals $a$ and its Y-axis equals $z$.
    On both figures, $a_1$ is the maximal observed IM value among ``non-failures'', and $a_2$ is the minimal one among ``failures''.
    \end{minipage}
    \caption{{Example of a type 3 data sample $(\mbf z^k,\mbf a^k)$ which gives a degenerate likelihood.}} %
    \label{fig:degenerate-frag}
\end{figure*}

\subsection{Constrained reference prior and posterior}\label{sec:prior}

The Bayesian framework considers $\theta$ to be stochastic. It permits the computation of the posterior distribution, i.e. the distribution of $\theta$ given the knowledge of the observations, providing \emph{a posteriori} estimates of the fragility curve and a quantification of estimation uncertainty.
However, the definition of the posterior necessitates to include a prior distribution on the parameter $\theta$. Indeed, denoting the posterior by $p(\theta|\mbf z^k,\mbf a^k)$, the Bayes theorem states 
    \begin{equation*}
        p(\theta|\mbf z^k,\mbf a^k) = \frac{\ell_k(\mbf z^k,\mbf a^k|\theta)\pi(\theta)}{\int_\Theta \ell_k(\mbf z^k,\mbf a^k|\tilde\theta)\pi(\tilde\theta)d\tilde\theta},
    \end{equation*}
with $\pi$ denoting the prior.

As mentioned in the introduction, different approaches exist in the literature in order to choose the prior. In this work, we take as a support the reference prior theory \cite{Bernardo1979a, VanBiesbroeckBA2023}, which provides tools to define a prior whose choice can be qualified as objective. 
{In this way, the methodology presented in this article has an “auditable” character, which can be essential for sensitive industries such as the nuclear industry.}
The reference prior $\pi^\ast$ is typically defined as a solution of the {optimization} problem $\pi^\ast\in\argmax_\pi\sI^k(\pi) $ when $k\to\infty$, where $\sI^k$ denotes the mutual information:
    \begin{equation*}
        \sI^k(\pi) = \EE_{\mbf z^{k},\mbf a^k}\big[D_\delta \big(\pi(\theta)||p(\theta|\mbf z^k,\mbf a^k)\big)\big],
    \end{equation*}
with $D_\delta$ being a $\delta$-divergence\footnote{Often, $\delta$-divergences are called $\alpha$-divergences. We keep the notation $\delta$ in this work.}: 
$D_\delta(p||q)=\int f_\delta\left(\frac{p(x)}{q(x)}\right)q(x)dx$ where $f_\delta(x)=\frac{x^\delta}{\delta(\delta-1)}$, $\delta\in(0,1)$\footnote{The usual definition of $f_\delta$ for $\delta$-divergences is $f_\delta(x)=\frac{x^\delta-\delta x-(1-\delta)}{\delta(\delta-1)}$. In our definition the other terms vanish.}. 
Maximizing $\sI^k$ amounts to maximizing the information brought by the observations to the posterior. 

The reference prior is in most {cases} the Jeffreys prior \cite{Clarke1994, VanBiesbroeckBA2023}, largely adopted by Bayesian statisticians for its property to be invariant by reparametrization of the {statistical model}.
{It has been derived and applied in the context of seismic fragility assessment of structures and components in previous works \cite{VanBiesbroeck2023, VanBiesbroeckUNCECOMP2023}. The authors showed that its use for the estimation of fragility curves outperforms current approaches found in the literature but, for practitioners, its implementation may suffer from some limitations because:}
\begin{itemize}
    \item[(i)] its theoretical expression does not have any known analytical form, and can be expensive to evaluate. We remind that the Jeffreys prior $J$ is defined by
        \begin{align*}
            J(\theta) &=\sqrt{|\det\cF(\theta)|},\\ \text{with}\quad \cF(\theta)&=-\sum_{z\in\{0,1\}}\int_{a\in\cA}\ell_1(z, a|\theta)\nabla^2_\theta\log\ell_1(z, a|\theta) da. 
        \end{align*}
    \item[(ii)] when the likelihood is degenerate (Definition \ref{def:degeneracy}), this prior yields an improper posterior which prevents its use in those cases.
\end{itemize}

The problem (ii) is of primary interest in our study.
Given the decay rates of the likelihood, it occurs with a wide range of non-informative priors. It is thus important to define the adequate requirement for a mathematically well-established \emph{a posteriori} estimation, namely a robust estimation of the fragility curve.

\begin{defi}[Robust \emph{a posteriori} estimation]\label{def:robust-estimation}
    The \emph{a posteriori} estimation is called \emph{robust} if the posterior is proper. {According to the results of the study of the asymptotic behavior of the posterior carried out by the authors in \cite{VanBiesbroeck2023}, for the probit-lognormal model, the posterior is proper if:}
    \begin{itemize}
        \item the likelihood is not degenerate and the prior verifies
        \begin{itemize}
        	\item[(i)] $\forall\beta>0,\, \pi(\theta)\aseq{\log\alpha\rightarrow\pm\infty}O(1)$,
        	\item[(ii)] $\forall\alpha>0,\,\pi(\theta)\aseq{\beta\rightarrow0}0(1) $, and $\pi(\theta)$ is intgrable in the neighborhood of $\beta\rightarrow\infty$.
        \end{itemize}
        \item the likelihood is degenerate, the prior is proper w.r.t $\beta$ and verifies $\forall\beta>0,\, \pi(\theta)\aseq{\log\alpha\rightarrow\pm\infty}O(1)$.
    \end{itemize}
\end{defi}

{Recall that the Jeffreys prior distribution leads to a proper posterior distribution when the likelihood is non-degenerate \cite{VanBiesbroeck2023}. This means that it satisfies the conditions just stated in the Definition \ref{def:robust-estimation}.} It is clear that a prior satisfying the second criterion of the definition  satisfies the first one as well. It is worth mentioning that Definition~\ref{def:robust-estimation} does not guarantee the absence of estimates that tend to unrealistic fragility curves
---such as unit step functions--- when the likelihood is degenerate.
{Indeed, since the likelihood is equal to $1$ when $\beta = 0$ in this case (see \cite{VanBiesbroeck2023} and Fig.~\ref{fig:degenerate-frag}),} only a ``very well-informed'' prior in the neighborhood of $\beta\to 0$ can avoid this phenomenon with certainty. However, as for IMs whose correlation with the structural response of interest tends towards 1, the value of $\beta$ of the associated fragility curve tends towards 0, the use of such an informed prior can, in return, penalize learning.

To avoid such a situation and taking as support the reference prior theory, the author defines  in \cite{VanBiesbroeck2024constraints} the minimal constraints which can be added to the optimization problem that defines the reference priors to ensure the solution is proper:
    \begin{equation}\label{eq:constrained-optim-problem}
        \pi^\ast\in\argmax_{\pi\ \ k\rightarrow\infty}\sI^k(\pi)\quad\text{subject to}\quad \int_\Theta g(\theta)\pi(\theta)d\theta<\infty.
    \end{equation}
{The subscript ``$k\rightarrow\infty$'' indicates here that the optimization problem is solved asymptotically as $k\rightarrow\infty$, see \cite{VanBiesbroeckBA2023} for more details about this formalism.}
The expression of the above constrained problem necessitates controlling the integrability of the Jeffreys prior by some function $g:\Theta\to(0,\infty)$, such that
    \begin{equation*}
         \int_\Theta J(\theta)g(\theta)^{1/\delta}d\theta <\infty \quad\text{and} \quad \int_\Theta J(\theta)g(\theta)^{1+1/\delta}d\theta <\infty.
    \end{equation*}
{The asymptotic decay rates of the Jeffreys prior are derived in our context in \cite{VanBiesbroeck2023}. We notice that} $g(\theta)=\beta^\eps$, with $\eps\in(0,\frac{2\delta}{1+\delta})$ satisfies the above {constraints}, so that the solution $\pi^\ast$ of Eq.~(\ref{eq:constrained-optim-problem}) is given by $\pi^\ast(\theta)\propto J(\theta)g(\theta)^{1/\delta}$. We can parametrize this prior w.r.t. $\gamma=\eps/\delta$. This yields a range of candidate priors to our problem:
    \begin{equation*}
        \pi^\ast_\gamma(\theta) = \frac{J(\theta)\beta^\gamma}{\int_\Theta J(\tilde\theta)\tilde\beta^\gamma d\tilde\theta},\quad\gamma\in\Big(0,\frac{2}{1+\delta}\Big)\subset(0,2).
    \end{equation*}

While $0<\gamma<2$, this prior satisfies both criteria required for a robust \emph{a posteriori} estimation (Definition \ref{def:robust-estimation}). The case $\gamma=0$ corresponds to the critical case where $\pi^\ast_\gamma$ equals the Jeffreys prior and thus provides a robust \emph{a posteriori} estimation if and only if the likelihood is not degenerate. The limit case $\gamma=2$ never provides a robust \emph{a posteriori} estimation. We precise that the methodology presented in the next section requires that the prior tackles degenerate-likelihood cases. The tuning of this hyper-parameter must be thought as the research for a balance between objectivity ($\gamma$ closer to $0$) and suitability for inference. The influence of $\gamma$ is studied in the practical application of our method in Section~\ref{sec:application}.

For a faster computation of the prior $\pi^\ast_\gamma$ (which is still built on the Jeffreys prior), 
the authors in \cite{VanBiesbroeck2023,Gu2019} suggest to approximate it from its decay rates.
The ones of $\pi^\ast_\gamma$ can be derived from the ones of Jeffreys, which are elucidated in \cite{VanBiesbroeck2023}. We thus suggest the following approximation of $\pi^\ast_\gamma$:
\begin{equation}\label{eq:tilde-pi-gamma}
    \tilde\pi^\ast_\gamma(\theta) \propto \frac{1}{\alpha(\beta^{1-\gamma}+\beta^{3-\gamma})} \exp\left(-\frac{(\log\alpha-\mu_A)^2}{2\sigma_A^2+2\beta^2}\right)
\end{equation}
where $\mu_A$ and $\sigma_A$ respectively denote the mean and the standard deviation of the r.v. $\log A$. We let the reader notice that the distribution of the IM {---which is approximated here by a lognormal distribution---} is involved in the prior expression. 
As a matter of fact, its presence within the likelihood expression (Eq.~(\ref{eq:likelihood})) propagates within the derivation of the Jeffreys prior, and so to its decay rates, as proven in \cite{VanBiesbroeck2023}.

The close connection that exists between the prior distribution of $\alpha$ and the one of $A$ emphasizes the need of taking into account a realistic distribution of the IM. {Approximating the IM distribution by a lognormal distribution is consistent with the distributions of the real and artificial signals of our study (see \cite{VanBiesbroeck2023,VanBiesbroeckUNCECOMP2023} and Fig.~\ref{fig:IM-density}). Some studies also show that a lognormal distribution is compatible with a scenario involving near-source ground motions \cite{Yamada2009}. In addition, this assumption is adopted for the development of ground motion prediction equations (GMPEs) \cite{Abrahamson1992,CB2014}, from which artificial signals can be generated \cite{ZENTNER2014}. The scope of application of the prior defined in Eq.~(\ref{eq:tilde-pi-gamma}) is therefore quite broad.}

Eventually, the prior and thus the posterior are known up to a multiplicative  constant. While Definition \ref{def:robust-estimation} is satisfied by the data and the prior, it is possible to generate estimates of $\theta$ from its posterior distribution through Markov Chain Monte Carlo methods.

\section{Sequential design of experiments%
}\label{sec:PEmethod}

\subsection{Methodology description}

The priors defined in Eq.~(\ref{eq:tilde-pi-gamma}) handle the degeneracy of the likelihood to provide robust \emph{a posteriori} estimation of the fragility curve. However, it remains important to reduce the occurrence of this phenomenon because: %
    \begin{enumerate}
        \item[(i)] {even if it partially disappears thanks to the slightly informed priors we suggest, the possibility of obtaining fragility curve estimates whose parameter $\beta\to 0$ (i.e., of obtaining unrealistic fragility curves) still exists with very small data sizes. As mentioned earlier, only a ``very well-informed'' prior distribution in the neighborhood of $\beta\to 0$ can avoid this phenomenon with certainty. However, its use is not desirable in order not to penalize the learning in cases where the effective value of $\beta$ is indeed very close to zero.}
        \item[(ii)] by nature, a likelihood becomes degenerate consequently to a lack of information within the observed data samples. Avoiding degeneracy should thus lead to a better understanding of the structure's response and its fragility curve.
    \end{enumerate}

In this paper, we propose to tackle degeneracy with an appropriate DoE strategy. Suppose that we have observed the sample $(\mbf z^k,\mbf a^k)$, we need to create a criterion to choose the next input $a_{k+1}$. 
Given our analysis (statement (ii) above) from an information theory viewpoint, the strategy to build the mentioned criterion is to maximize the impact that the selected observation would have on the posterior distribution. 
This strategy echoes the one that supports the reference prior definition as a maximal argument of the mutual information, so that we suggest a similar criterion to select $a_{k+1}$:
      \begin{align}\label{eq:index}
        &a_{k+1}=\argmax_{a\in\cA}\cI_{k+1}(a);\\ \quad 
         &   \cI_{k+1}(a_{k+1}) = \EE_{{z_{k+1}}|\mbf a^{k+1},\mbf z^k}[D_\delta(p(\theta|\mbf z^k,\mbf a^k)||p(\theta|\mbf z^{k+1},\mbf a^{k+1}))].\nonumber
    \end{align}
The index $\cI_{k+1}$ can be seen as a sensitivity index measuring the sensitivity of the posterior w.r.t. the data \cite{DaVeiga2015}. 
Its sequential maximization amounts to maximize the impact of the data to the posterior.
Our index can be seen as derived from the popular framework of Stepwise Uncertainty Reduction techniques \cite{Bect2018}. 
Typically,  those methods are formulated with an \emph{a posteriori} variance within the expectation in Eq.~(\ref{eq:index}) instead of the dissimilarity measure suggested here.
Algorithm \ref{alg:PE} presents a practical pseudo-code of our methodology. It requires to derive an approximation of the index $\cI_{k+1}$, which we elucidate in the following section.

\renewcommand{\algorithmicrequire}{\textbf{Notations:}}

 \begin{algorithm*}
		\caption{Design of experiments}
		\begin{algorithmic}
            \REQUIRE \begin{tabular}[t]{l}
			Seismic signal : $\cS$\\ Intensity measure of the seismic signal: $\mathrm{IM}(\cS)$\\ Mechanical response to the seismic signal (failure or success): $\mathrm R(\cS)$ 
			\end{tabular}\renewcommand{\algorithmicrequire}{\textbf{Initialization:}}
            \REQUIRE \begin{tabular}[t]{l}
			$k_0>0$ (in practice $k_0=2$), initial seismic signals $\cS_1,\dots,\cS_{k_0}$\\
            Define the initial data: $\mbf z^{k_0}=(\mathrm{R}(\cS_1),\dots,\mathrm{R}(\cS_{k_0})),\   \mbf a^{k_0}=(\mathrm{IM}(\cS_1),\dots,\mathrm{IM}(\cS_{k_0}))$
			\end{tabular}
            \FOR{$k=k_0\dots k_{\text{max}}-1$ }
	            \STATE Approximate $\cI_{k+1}$ via Monte Carlo sampling
                \STATE Compute $a_{k+1}=\argmax_a\cI_{k+1}(a)$
	            \STATE Choose a seismic signal $\cS_{k+1}$ such that $\mathrm{IM}(\cS_{k+1})=a_{k+1}$
                \STATE {Perform the %
                experiment and} define $z_{k+1}=\mathrm{R}(\cS_{k+1})$
	        \ENDFOR
		\end{algorithmic}
	    \label{alg:PE}
	\end{algorithm*}

\subsection{Approximation of the index}

Let us adopt the short notation
    \begin{equation*}
        \Psi^z_a(\theta)=\Phi\left(\frac{\log a-\log\alpha}{\beta}\right)^z\left(1-\Phi\left(\frac{\log a-\log\alpha}{\beta}\right)\right)^{1-z}.
    \end{equation*}
The index to maximize has the following form:
    \begin{align*}
        &\cI_{k+1}(a_{k+1})\\ &= (\delta(\delta-1))^{-1}
            \sum_{z\in\{0,1\}}
                \PP({z_{k+1}}=z|\mbf a^{k+1},\mbf z^{k})
                \Delta_{k+1}(\mbf a^{k+1},\mbf z^{k+1})
    \end{align*}
with 
        $\PP({z_{k+1}}=z|\mbf a^{k+1},\mbf z^{k+1}) = \int_{\Theta}\Psi^{z}_{a_{k+1}}(\theta)p(\theta|\mbf z^{k},\mbf a^{k})d\theta, $
and
    \begin{align*}
        &\Delta_{k+1}(\mbf a^{k+1},\mbf z^{k+1})\\ &=
            \int_{\Theta}\left(\frac{p(\theta|\mbf z^{k},\mbf a^{k})}{p(\theta|\mbf z^{k+1},\mbf a^{k+1})}\right)^{\delta}p(\theta|\mbf z^{k+1},\mbf a^{k+1})d\theta \nonumber\\
        &=\int_\Theta \left(\frac{\int_\Theta p(\vartheta|\mbf z^k,\mbf a^k)\Psi^{z_{k+1}}_{a_{k+1}}(\vartheta)d\vartheta}{\Psi^{z_{k+1}}_{a_{k+1}}(\theta) } \right)^\delta \frac{\Psi^{z_{k+1}}_{a_{k+1}}(\theta)p(\theta|\mbf z^k,\mbf a^k) }{\int_\Theta p(\vartheta|\mbf z^k,\mbf a^k)\Psi^{z_{k+1}}_{a_{k+1}}(\vartheta)d\vartheta} d\theta\\
        &= \left(\int_\Theta \Psi^{z_{k+1}}_{a_{k+1}}(\theta) p(\theta| \mbf z^k, \mbf a^k) d\theta\right)^{\delta-1}  \int_\Theta \Psi^{z_{k+1}}_{a_{k+1}}(\theta)^{1-\delta} p(\theta|\mbf z^k,\mbf a^k)  d\theta.\nonumber
    \end{align*}
The index can thus be approximated via Monte-Carlo from a sample of $\theta$ distributed according to a preceding posterior distribution {$p(\theta|\mbf z^q,\mbf a^q)$, with $q\leq k$; $q$ can equal $0$ in which case the sample is distributed w.r.t. the prior. In practice, it is suitable to set $q=k$. One can rely on the following formulas:}
    \begin{align*}
        &\int_\Theta\Psi^{z_{k+1}}_{a_{k+1}}(\theta)p(\theta|\mbf z^k,\mbf a^k)d\theta\\ &= \int_\Theta \Psi^{z_{k+1}}_{a_{k+1}}(\theta)\prod_{j=q+1}^k\Psi^{z_{j}}_{a_{j}}(\theta)p(\theta|\mbf z^q,\mbf a^q)\frac{1}{L_q^k}d\theta ,\\
   \text{and}\quad     &\int_\Theta\Psi^{z_{k+1}}_{a_{k+1}}(\theta)^{1-\delta}p(\theta|\mbf z^k,\mbf a^k)d\theta\\ &= \int_\Theta \Psi^{z_{k+1}}_{a_{k+1}}(\theta)^{1-\delta}\prod_{j=q+1}^k\Psi^{z_{j}}_{a_{j}}(\theta)p(\theta|\mbf z^q,\mbf a^q)\frac{1}{L^k_q}d\theta.\nonumber
    \end{align*}
where $L^k_q = \int_\Theta \prod_{j=q+1}^k\Psi^{z_{j}}_{a_{j}}(\theta)p(\theta|\mbf z^q,\mbf a^q)d\theta$ does not depend on $a_{k+1}$.
In this way, if $\theta_1,\dots,\theta_M$ is an independent and identically distributed sample according to the posterior distribution $p(\theta|\mbf z^q,\mbf a^q)$, we can define for $\zeta=1$ and $\zeta=1-\delta$:
{%
    \begin{align*}
        Q_\zeta^{l = 0,1} &= \frac{1}{M}\sum_{i=1}^M\Psi^{l = 0,1}_{a_{k+1}}(\theta_i)^{\zeta}\prod_{j=q+1}^k\Psi^{z_{j}}_{a_{j}}(\theta_i),
    \end{align*}
}%
to approximate easily  $\cI_{k+1}(a_{k+1})$ up to the constant $(L^k_q)^{\delta+1}$:
    \begin{align*}
        \cI_{k+1}(a_{k+1})&\simeq (L^k_q)^{-\delta-1}  %
        (\delta(\delta-1))^{-1} \left((Q^0_1)^{\delta} Q^0_{1-\delta} + (Q^1_1)^{\delta} Q^1_{1-\delta}\right).%
    \end{align*}
The above being proportional to $(\delta(\delta-1))^{-1} \left((Q^0_1)^{\delta} Q^0_{1-\delta} + (Q^1_1)^{\delta} Q^1_{1-\delta}\right)$.
Eventually, $a_{k+1}$ is chosen as the maximal argument of the right-hand side term in the above equation.

\subsection{{Stopping} criterion \label{sec:stopping_crit}}

Algorithm~\ref{alg:PE} consists in a loop where $k$ iterates from $k_0$ to $k_{\max}-1$. The upper value $k_{\max}$ represents the total number of experiments at the end of the campaign. 
In practical studies, that value is limited by the cost of experiments.
As a matter of fact, additional experiments are expected to provide a quantity of information that enhances the estimates. %
Practitioners could judge to cease the campaign if the cost of an experiment overtakes its benefits.

Such quantity of information provided by data samples is derived through our method in the index $\cI$, that is why we suggest to study its variation to elucidate a {stopping} criterion:
\begin{equation*}
    \cV\cI_k = \frac{|\cI_{k+1}(a_{k+1})-\cI_k(a_k)|}{|\cI_k(a_k)|}.
\end{equation*}
When the index $\cV\cI_k$ falls below a certain threshold value, the method has ceased to leverage enough from the observations.
This index can be appreciated alongside the variation of the estimated fragility curve itself:
    \begin{equation*}
            \cV\cP_k = \frac{\|m^{|\mbf z^k,\mbf a^k} - m^{|\mbf z^{k+1},\mbf a^{k+1}}\|_{L^2}}{\|m^{|\mbf z^k,\mbf a^k}\|_{L^2}},
    \end{equation*}
where $m^{|\mbf z^k,\mbf a^k}$ is the median of the fragility curve estimate given the observations $(\mbf z^k,\mbf a^k)$; details concerning the practical definition of the norm $\|\cdot\|_{L^2}$ are given in Section~\ref{sec:metrics}.
When the index $\cV\cP_k $ falls, the estimated fragility curve given by the method has stopped to distinctly evolve.

The index $\cV\cP_k$ constitutes a criterion which is more perceptible {for practitioners}.
In the next section where we apply our method to a practical case study ,{ we verify that $\cV\cI_k$ and $\cV\cP_k$ give pragmatic and consistent information.}

\section{Application of the DoE methodology}\label{sec:application}

{This section is devoted to the application of our methodology to a practical case from the nuclear industry. A comprehensive study of the performance of the method is proposed by considering the two classical IMs that are the PGA and the PSA. This study is further supported by the study of a toy case that is presented in \ref{app:toycases}.}

    \subsection{Case study presentation}

        \subsubsection{Seismic signals generation}\label{sec:seismic-generator}
    
        {The objective of this work being to propose a strategy for designing experiments, it is necessary to rely on a large database of seismic signals to be able to choose them optimally. $10^5$ artificial seismic signals were therefore generated.} They were obtained with the SSG proposed by~\citet{Rezaeian2010}. This generator implemented in~\cite{Sainct2020} was calibrated from 97 real accelerograms selected in the European Strong Motion Database for a magnitude $M$ such that $5.5 \leq M \leq 6.5$, and a source-to-site distance $R < 20$~km~\cite{ESMD}. {From these artificial signals, two different IMs were derived for illustration purposes. These are, on the one hand, the PGA and, on the other hand, the PSA at 5~Hz for a damping ratio of 1\% (in accordance with the equipment's characteristics, see section~\ref{sec:ASG}). In Fig.~\ref{fig:IM-density} the distributions of these two IMs are compared with the ones that come from the 97 real seismic signals. These comparisons show that the synthetic signals have realistic features.}

        \begin{figure*}
            \centering
            \includegraphics[width=5cm]{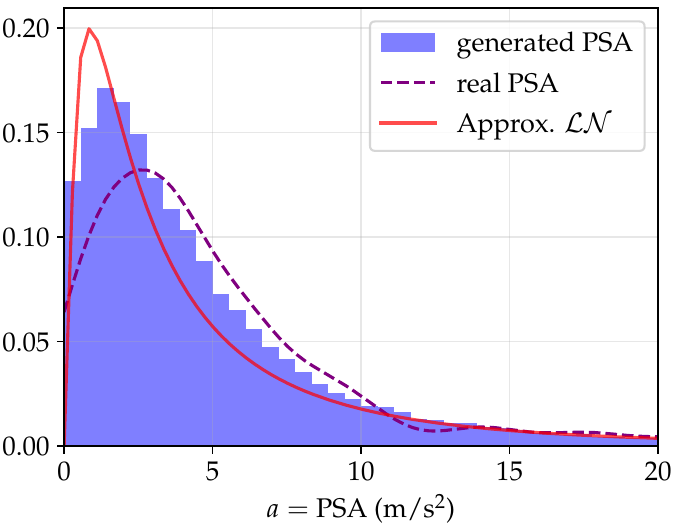}
            \includegraphics[width=5cm]{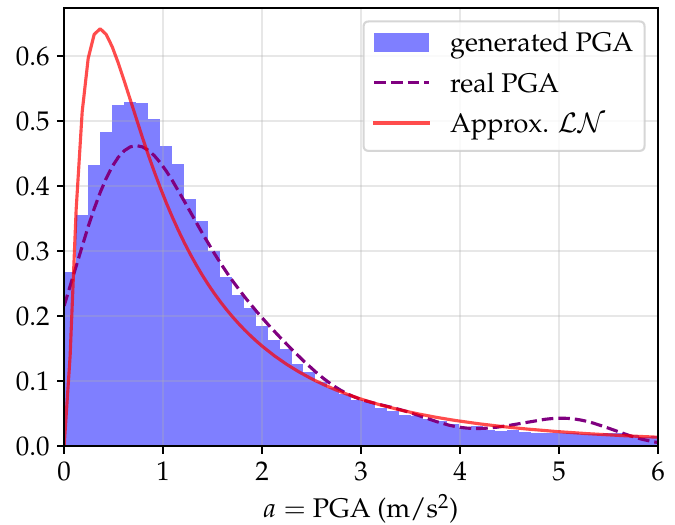}\\[0.7em]
            \begin{minipage}{0.94\textwidth}\footnotesize 
                \itshape Histograms (in blue) of IMs (PSA in left figure and PGA in right figure) derived from $10^5$ generated synthetic signals. They are compared with the densities of IMs coming from real accelerograms estimated by Gaussian kernel estimation (dashed lines), and with the densities of lognormal distributions with same medians and same log deviations (solid lines).
            \end{minipage}
            \caption{Statistical distributions of the PSA and the PGA.}
            \label{fig:IM-density}
        \end{figure*}

        \subsubsection{Description of the mechanical equipment \label{sec:ASG}}

        {In this case study, we investigate the seismic fragility of a piping system {that is part of a French pressurized water reactor} and that was tested on the Azalee shaking table at the EMSI laboratory of CEA/Saclay, as shown in Fig.~\ref{fig:ASG}-left. Fig.~\ref{fig:ASG}-right depicts the finite element model (FEM), based on beam elements and implemented through the proprietary FE code CAST3M~\cite{CAST3M}. The validation of the FEM was carried out thanks to an experimental campaign described in~\cite{TOUBOUL1999}.}

        The mock-up comprises a carbon steel TU42C pipe with an outer diameter of 114.3 mm, a thickness of 8.56 mm, and a 0.47 elbow characteristic parameter. This pipe, filled with water without pressure, includes three elbows, with a valve-mimicking mass of 120 kg, constituting over 30\% of the mock-up's total mass. One end of the mock-up is clamped, while the other is guided to restrict displacements in the X and Y directions. Additionally, a rod is positioned atop the specimen to limit mass displacements in the Z direction (refer to Fig.~\ref{fig:ASG}-right). During testing, excitation was applied exclusively in the X direction.

        \begin{figure}[H]
		\centering		
		\includegraphics[width=4.5cm]{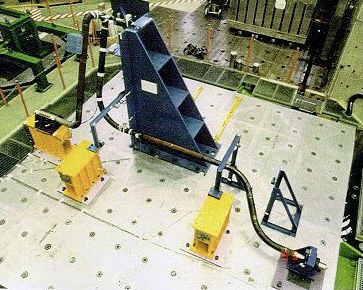}
		\hspace{0.5cm}
		\includegraphics[width=2cm]{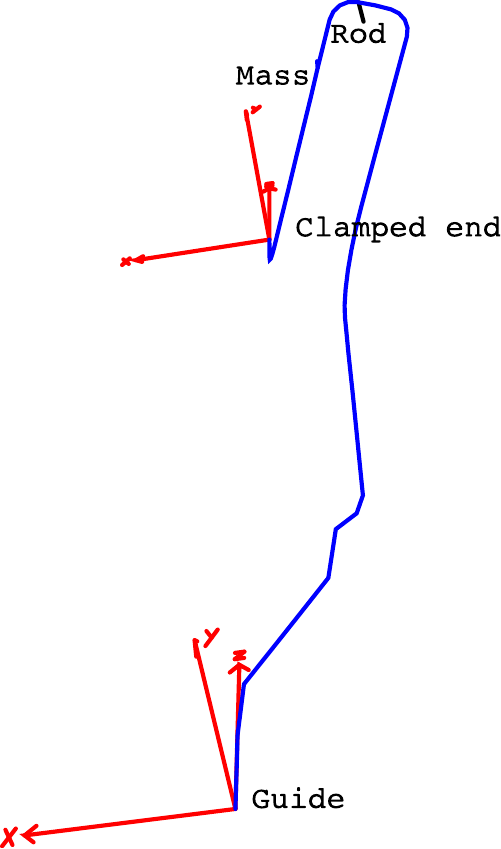} %
        \caption{Overview of the piping system.}
		\label{fig:ASG}
	\end{figure}

        In order to conduct comparative performance studies, numerous simulations have been performed. They were carried out from a subset of $8\cdot10^4$ of the $10^5$ artificial seismic signals. Nevertheless, as in practice the piping system is located in a building, the artificial signals were filtered using a fictitious 2\% damped linear single-mode building at 5 Hz, which corresponds to the first eigenfrequency of the 1\% damped piping system.  In such a situation, for some seismic signals, the behavior of the piping system is nonlinear. Regarding the nonlinear constitutive law of the material, a bilinear law exhibiting kinematic hardening was used to reproduce the overall nonlinear behaviour of the mock-up with satisfactory agreement compared to the results of the seismic tests~\cite{TOUBOUL1999}. Following the recommendation in~\cite{TOUBOUL2006}, we consider that the EDP is the out-of-plane rotation of the elbow near the clamped end of the mock-up. As a result we have a dataset of $8\cdot10^4$ independent tuples of the form (PSA, PGA, EDP). In Fig.~\ref{fig:scattersIMs} are plotted the elements of the datasets (PSA, EDP) and (PGA, EDP), along with different critical thresholds.
        The binary data are defined from the condition that failure happens when the EDP is larger than the threshold value. {It should be remembered that obtaining these results represents approximately 90 days of sequential calculations on a simple laptop. However, such a computational effort cannot be implemented in practice on the scale of an industrial installation where several of its equipment must be evaluated. These results allow us to estimate reference, PGA-based and PSA-based, fragility curves for this case study and allow us to evaluate our methodology.}

        \begin{figure*}
        \centering%
        \includegraphics[width=5cm]{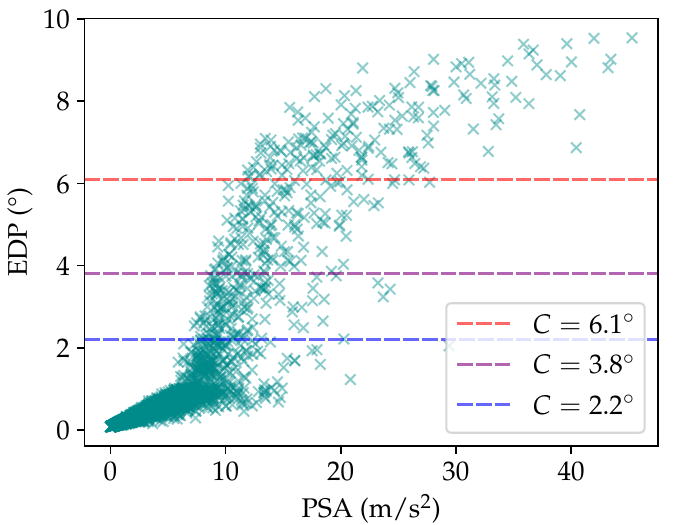}%
        \includegraphics[width=5cm]{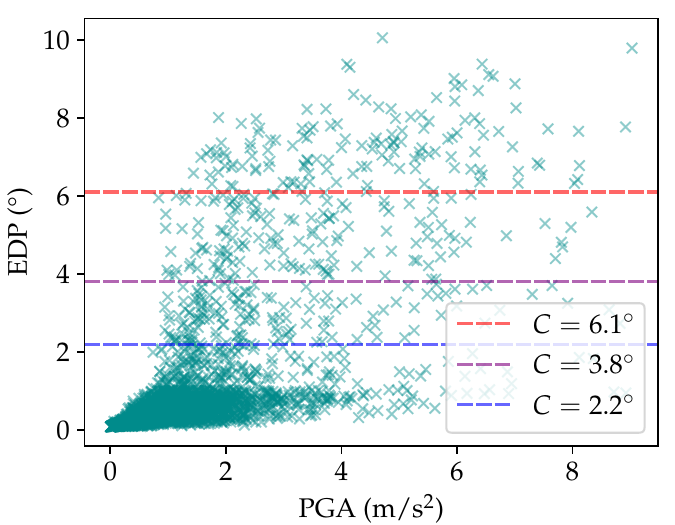}%
        \caption{Results of the $8\cdot10^4$ numerical simulations. Each cross is an element of the dataset (IM, EDP) where the IM is the PSA (left) and the PGA (right). Different critical rotation thresholds $C$ are plotted in dashed lines. They yield different proportions of failures in the dataset: respectively 95$\%$ (red), $90\%$ (purple) and $85\%$ (blue).}
        \label{fig:scattersIMs}
        \end{figure*}
        
        \subsubsection{Reference fragility curves}\label{sec:reference}

        The complete set of $8\cdot10^4$ simulations provides a satisfactory dataset to have a reference of the seismic fragility curve that we aim to estimate. 
        For a robust non-parametric reference, it is possible to derive Monte Carlo estimators of local probability of failures w.r.t. the IM. \citet{TREVLOPOULOS2019} suggest to take clusters of the IM using K-means as the available IM are not uniformly distributed. Such reference comes along with a confidence interval as the dataset is not infinite. In Fig.~\ref{fig:reference-frags}, we compare this non-parametric reference with the parametric fragility curve using the probit-lognormal model (Eq.~(\ref{eq:probit-model})) where $\alpha$ and $\beta$ are estimated by MLE using the full dataset as well. {This comparison is presented considering different critical rotation thresholds $C$ of the equipment for both the PSA and the PGA. First of all, it should be noted that, with the PGA as IM, it is not possible to completely describe the fragility curve. For the maximum PGA values observed, the failure probabilities stagnate between 0.5 and 0.8 depending on the failure criterion considered. Therefore, the PGA is not the most suitable IM of the two. For the type of structure considered here, this point is well documented in the literature. Then, the comparisons demonstrate, in our setting, the adequacy of a probit-lognormal modeling of the fragility curves, even with high failure thresholds.} Its bias with non-parametric fragility curve exists but remains limited. When the number of observations is small, this model bias is negligible in front of the uncertainty on the estimates.
        
        \begin{figure*}
        \centering
        \includegraphics[width=5cm]{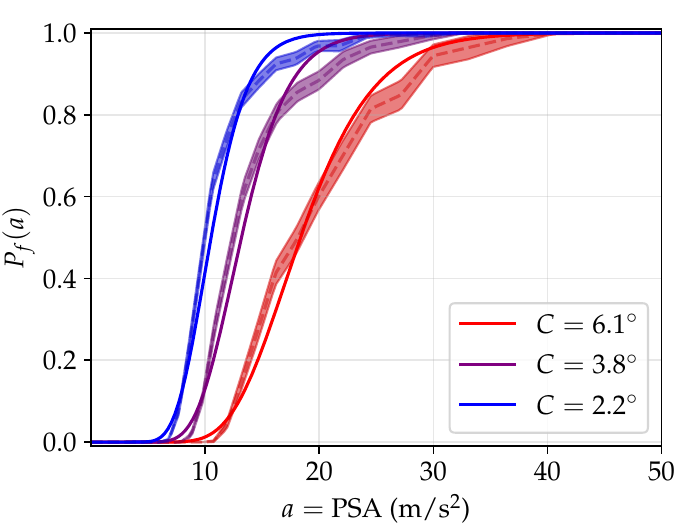}
        \includegraphics[width=5cm]{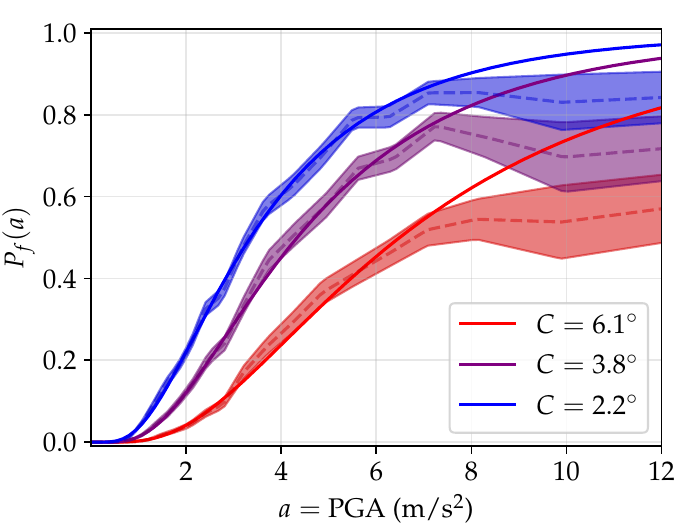}
        \caption{{Reference non-parametric fragility curves obtained via Monte Carlo estimates (dashed lines) surrounded by their $95\%$ confidence intervals, for different critical rotation threshold $C$ with (left) the PSA and (right) the PGA as IM.} The thresholds yield different proportions of failures in the dataset: respectively $95\%$ (red), $90\%$ (purple) and $85\%$ (blue).
        For each value of $C$ are plotted (same color, solid line) the corresponding probit-lognormal MLE.}
        \label{fig:reference-frags}
        \end{figure*}

        {In the sequel, although the PGA is not the most suitable IM for our problem, we present the results of applying our methodology with both the PSA and the PGA, for the sake of completeness. However, and without loss of generality, the critical rotation threshold is set at $C = 3.8^{\circ}$. This represents the 90\%-level quantile derived from the sample of $8\cdot10^4$ nonlinear numerical simulations.}

        \subsubsection{{Sequential design of experiments} {implementation}}\label{sec:SDoE}

        %
        %

        %
        {The sequential DoE methodology presented in Algorithm \ref{alg:PE} involves (i) generating new seismic signals from a precise value of their IMs and (ii) performing the associated numerical simulations or experimental tests, depending on the context. In this case, we rely on the large database that we have built up for comparative performance studies, but that does not call into question the process of selection of the ``optimal'' signals. It simply avoids having to perform new numerical simulations.}
        
        Let us denote $\cS_1,\dots,\cS_N$, $N=8\cdot10^4$ the synthetic signals that constitutes the aforementioned dataset. If
        $\cS_{i_1},\dots,\cS_{i_k}$ are the ones from which result the current observations in our method, and if we look for the generation of $\cS_{i_{k+1}}$ such that its IM equals $a_{k+1}$,
        then we select $\cS_{i_{k+1}}$ such that
            \begin{equation*}
                i_{k+1} = \argmin_{\substack{i\in\{0,\dots,N\}\\ i\not\in\{i_1,\dots,i_k\}}} |a_{k+1}-\mathrm{IM}(\cS_{i})|;
            \end{equation*}
        where $\mathrm{IM}(\cS_{i})$ denotes the IM of signal $\cS_{i}$.

    \subsection{{Fragility curves estimations, benchmarking metrics and model 
    bias}}\label{sec:metrics}
    
    \subsubsection{{Fragility curves estimations}}\label{sec:estimations}

    Estimations of seismic fragility curves given observation $(\mbf z^k,\mbf a^k)$ are obtained by sampling independent and identically distributed values of $\theta$ from the posterior distribution $p(\theta|\mbf z^k,\mbf a^k)$. That sampling can be done via MCMC methods. In our work we use an adaptive Metropolis-Hastings algorithm \cite{Haario2001} which necessitates iterative evaluations of the posterior up to a constant.
    For a given value of $\gamma$, we can evaluate the posterior from the approximation $\tilde\pi^\ast_\gamma$ suggested in Eq.~(\ref{eq:tilde-pi-gamma}). In \cite{VanBiesbroeck2023}, a more exact but more expensive computation of the Jeffreys prior is proposed. We precise that we have verified that the results obtained from their method to compute $\pi^\ast_\gamma$ or from 
    the approximation $\tilde\pi^\ast_\gamma$ are indiscernible. Their comparison is not the point of this study and is not discussed in the following.

    \subsubsection{{Benchmarking metrics}}\label{sec:Bmetrics}
    
    To evaluate our estimates and the performances of our methodology, we propose to define metrics on the \emph{a posteriori} fragility curves that is expressed conditionally to $\theta$.
    The latter can be defined as the
    random process $a\mapsto P_f^{|\mbf z^k,\mbf a^k}(a)$ where  $P_f^{|\mbf z^k,\mbf a^k}(a)=\Phi(\beta^{-1}\log a/\alpha)$. It has a distribution that naturally inherits from the posterior distribution of $\theta$.
    For each value of $a$, we note $q_r^{|\mbf z^k,\mbf a^k}(a)$ the $r$-quantile of $P_f(a)$, and $m^{|\mbf z^k,\mbf a^k}(a)$ its median.
    These are defined for any $a\in\cA=[0,a_{\text{max}}]$, where $a_{\text{max}}$ corresponds to the highest value of the IM that exists in the dataset described in Section~\ref{sec:seismic-generator} ($a_{\text{max}}=12$m/s$^2$ for the PGA, and  $a_{\text{max}}=60$m/s$^2$ for the PSA).
    We define:
    \begin{itemize}
        \item The square bias to the median: $\cB^{|\mbf z^k,\mbf a^k} = \|m^{|\mbf z^k,\mbf a^k}- P_f^{\mathrm{ref}}\|^2_{L^2}$, where %
        $P^{\mathrm{ref}}_f$ denotes the reference non-parametric fragility curve computed as described in Section~\ref{sec:reference}.
        \item The quadratic error: $\cE^{|\mbf z^k,\mbf a^k}=\EE_{\theta|\mbf z^k,\mbf a^k}\left[\|P_f^{|\mbf z^k,\mbf a^k}-P_f^{\mathrm{ref}} \|^2_{L^2}\right].$
        \item The $1-r$ square credibility width: $\cW^{|\mbf z^k,\mbf a^k}= \|q_{1-{r/2}}^{|\mbf z^k,\mbf a^k} - q_{{r/2}}^{|\mbf z^k,\mbf a^k}\|_{L^2}^2$.
    \end{itemize}
    
    The norm $\|\cdot\|_{L^2}$ is defined by 
        \begin{equation}
            \|P\|_{L^2}^2 = \frac{1}{\tilde a_{\max}-\tilde a_{\min}}\int_{\tilde a_{\min}}^{\tilde a_{\max}}P(a)^2da. \nonumber
        \end{equation}
    For a consistent quantification of the above errors no matter the selected IM, $\tilde a_{\min}$ and $\tilde a_{\max}$ are selected to match a quantile of the reference fragility curve:
    we fix $0<q_1<q_2<1$ and choose $\tilde a_{\min}$ and $\tilde a_{\max}$ such that $q_1=P^{\text{ref}}_f(\tilde a_{\min})$ and $q_2=P^{\text{ref}}_f(\tilde a_{\max})$. In our work, the quantities defined in this section are derived considering the PGA and the PSA as IMs.
    The domain of the reference fragility curve is more limited in the case of the PGA (see Fig.~\ref{fig:reference-frags}): in this case the reference probability of failure lives between around $10^{-3} $ and $0.9$.
    Thus, we fixed the values of $q_1$ and $q_2$ to match these in any case. The implementation of these metrics is done through Monte-Carlo derivation and numerical approximation of the integrals from Simpsons' interpolation on a sub-division of $[\tilde a_{\min},\tilde a_{\text{max}}]$.

    \subsubsection{{Model bias}}\label{sec:modelbias}

    {In practice, the probit-lognormal model is susceptible to be biased. So, one of the metrics for evaluating its performance is to assess the square bias to the median, $\cB^{|\mbf z^k,\mbf a^k}$, defined in the previous section. To facilitate the interpretation of the values provided by this metric, we propose estimating the square model bias, $\cM\cB$, which corresponds to the deviations between $P_f^{\text{ref}}$ and $P_f^{\text{MLE}}$. It can be visualized in Fig.~\ref{fig:reference-frags} and is defined by:
        \begin{equation}
           \cM\cB = \|P_f^{\text{ref}}-P_f^{\text{MLE}}\|_{L^2}^2. \label{eq:MB}
        \end{equation}   
    $P_f^{\text{MLE}}$ denotes a probit-lognormal curve whose parameter $\theta$ is obtained by MLE given the full batch of $8\cdot10^4$ simulations. {Obviously, this curve would not be known in a real situation since its knowledge is equivalent to the knowledge of $P_f^{\text{ref}}$. It serves the assessment of our method: $P_f^{\text{MLE}}$ is expected to approach the ``best possible'' curve among probit-lognormal ones for this case study.} 
     By construction, the model bias is influenced by the distribution of the IM since $P_f^{\text{MLE}}$ is estimated by considering all the available samples. We cannot therefore speculate whether this is the absolute minimum bias and its value is only presented for information purposes. We return to this point in the section devoted to the interpretation of the results.}
    
    \subsection{Numerical results}\label{sec:numerical-results}

    {This section is devoted to the presentation of the numerical results first from a qualitative point of view, then from a quantitative one. To have an overall view of the performance of our methodology, we compare the results of (i) the so-called ``standard approach'' which consists in sampling the IM of interest according to its standard distribution considering the Jeffreys prior (i.e., the prior defined in Eq.~(\ref{eq:tilde-pi-gamma}) for $\gamma$ equal to 0) and (ii) the DoE approach which has been implemented for different values of $\gamma$.
    
    It should be noted that we only make comparisons with the so-called standard approach because it has been demonstrated in previous works the superiority of this approach compared to MLE \cite{VanBiesbroeck2023,VanBiesbroeckUNCECOMP2023}.}

    \begin{figure*}
        \centering%
        \includegraphics[width=5cm]{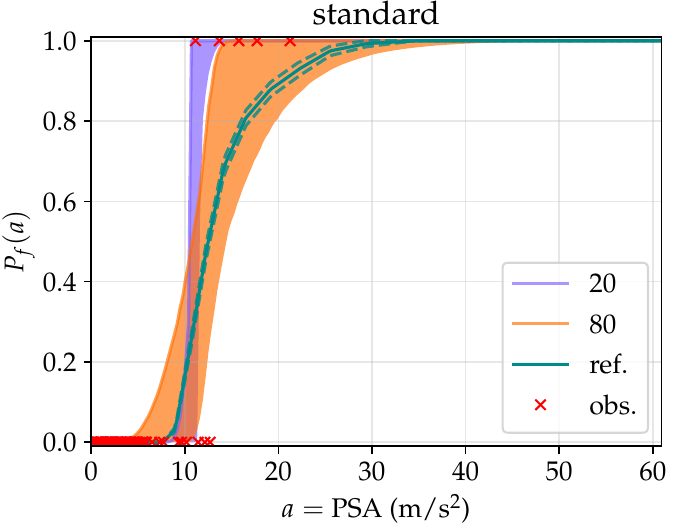}%
        \includegraphics[width=5cm]{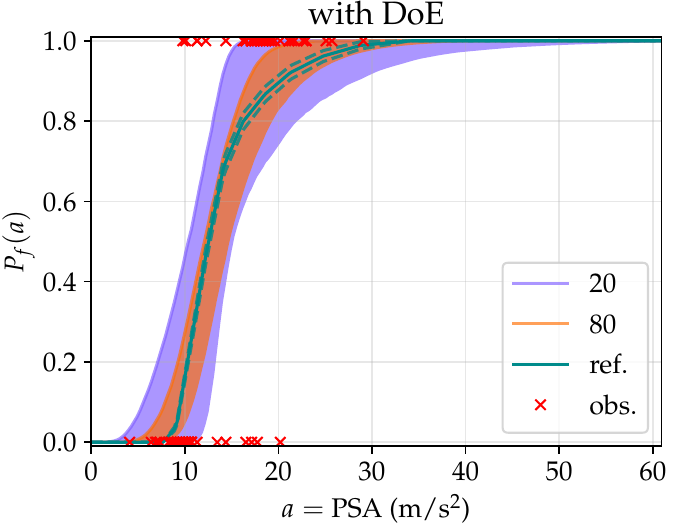}\\[5pt]
        \includegraphics[width=5cm]{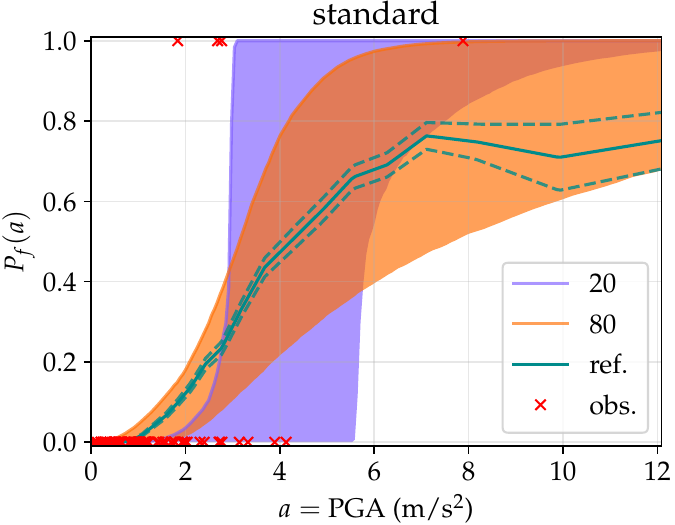}\includegraphics[width=5cm]{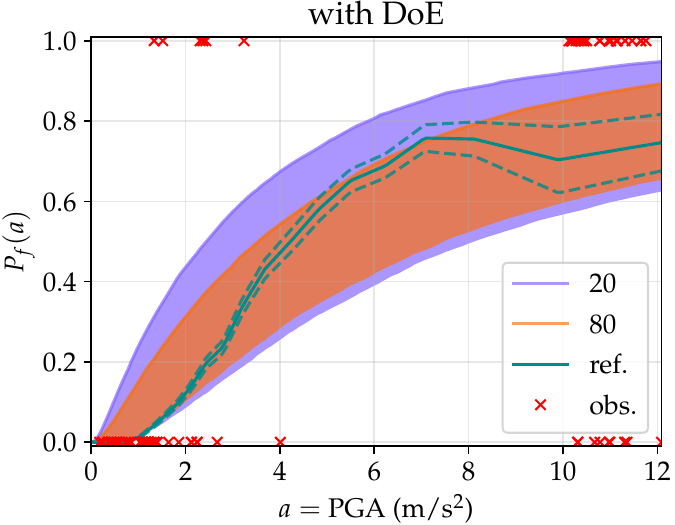}\\[0.7em]%
        \begin{minipage}{0.94\textwidth}
            \itshape\footnotesize 
            Top figures: the IM is the PSA; bottom figures: the IM is the PGA. Left figures: the seismic signals are chosen w.r.t. their standard distribution. Right figures: they are chosen w.r.t. our DoE strategy.
            On each figure: credibility intervals from $20$ (purple) or $80$ (orange) observations, reference fragility curves (green solid line) from Monte Carlo simulations and their {$95\%$}-confidence intervals (green dashed lines). The red crosses represent the $80$ observations.
        \end{minipage}
        \caption{Examples of fragility curves estimation.}
        %
        \label{fig:ex-estfrag}
    \end{figure*}

    \subsubsection{Qualitative study}
    
    In Fig.~\ref{fig:ex-estfrag} we present examples of \emph{a posteriori} fragility curve estimations. They take the form of {$95\%$-credibility intervals. These qualitative graphs allow us to appreciate the performance of the estimation when the IMs are selected using our DoE method compared to the standard approach. In this example, the DoE was implemented by setting $\gamma=0.5$.} The same comments stay valid for any of the two IMs we have taken into account in this study ---the PSA and the PGA---, namely: 
    \begin{enumerate}
        \item[(i)] When the number of observed data is really small (20 samples), the standard method tends to provide ``vertical'' fragility curves estimates, which result in ``vertical'' credibility intervals. By verticality, understand here that the tangent of the estimated curve at the median has an infinite slope (i.e., $\beta=0$). This phenomenon is a consequence of the fact that the likelihood is degenerate in this case.
        The credibility intervals are thus tighter but strongly incorrect.
        When the DoE is implemented with the same number of samples, the phenomenon fades and the credibility intervals follow accurately the reference curve.
        \item[(ii)] When a more decent number of data are observed (80 samples), the credibility intervals given from the DoE follow quite closely the reference curve. Within the observed seismic signals, we then notice a better balance between the ones with large IMs and the ones with small IMs compared to the standard method {(see the ``red crosses'' in the figures).} The consequence is a thinner credibility interval given by the DoE method.
    \end{enumerate}

\subsubsection{Quantitative study}\label{sec:quantitative-study}

Figures \ref{fig:errors-psa} to \ref{fig:variaP} present 
    quantitative results that go beyond a single example.
    Within those, 12 methods are compared:
        the standard method and the DoE methods for any $\gamma\in\{0,$ $0.1,$ $0.3,\dots,1.9\}$.
    The case $\gamma=0$ is particular: in this case the prior $\pi^\ast_\gamma$ does not always satisfy the criteria for a robust \emph{a posteriori} estimation (Definition~\ref{def:robust-estimation}) even though that is essential for the DoE to be implemented. Thus, $\gamma=0$ corresponds to an alteration of the DoE method:
        first, data are collected given the DoE method carried out with $\gamma=0.1$; second, the prior $\pi^\ast_{0.1}$ is replaced with $\pi^\ast_0$ to compute the estimates.
    This particular treatment is carried out in order to quantify how that parametrized constraint affects the estimates.

    {For each of these methods, we have carried out $100$ replications of the same numerical experiment, namely: (i) generate an observed sample of $k_{\max}=250$ data items and (ii) derive for any $k=10,$ $20,\dots,k_{\max}$ the posterior distribution $p(\theta|\mbf z^k,\mbf a^k)$.
    These replications provide evaluations of the mean values of the metrics defined in Section~\ref{sec:metrics}. The results are shown in Fig.~\ref{fig:errors-psa} for the PSA and in Fig.~\ref{fig:errors-pga} for the PGA.}

    \paragraph{{Overall comments : influence of $\gamma$}}
    
    {For both IMs, Figures~\ref{fig:errors-psa} and \ref{fig:errors-pga} show that the discrepancies between all the DoE-based results are negligible. We do not clearly distinguish a hierarchy of these with respect to the value of $\gamma$.}
    We have verified that the selected IMs by the DoE methods are actually poorly influenced by the value of $\gamma$ and that any change of performance impacted by the value of $\gamma$ is hard to distinguish. We conclude that its influence on the posterior estimation remains small. Thus, our constrained reference prior is close enough to the unconstrained ($\gamma = 0$) one and conserves its ``objectivity'' qualification.

\begin{figure*}
    \centering%
    \includegraphics[width=5cm]{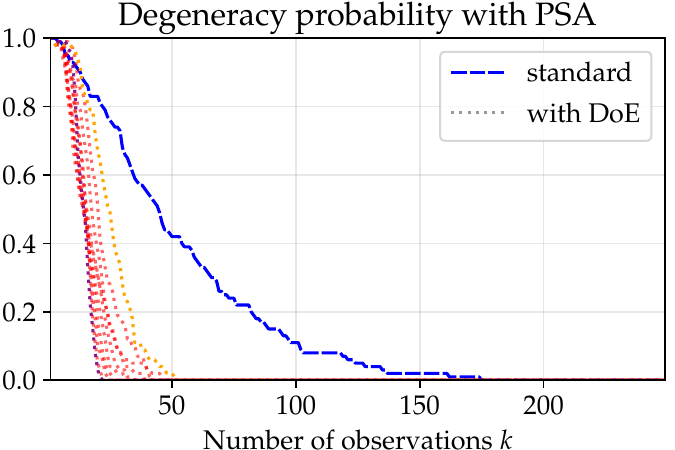}\ 
    \includegraphics[width=5cm]{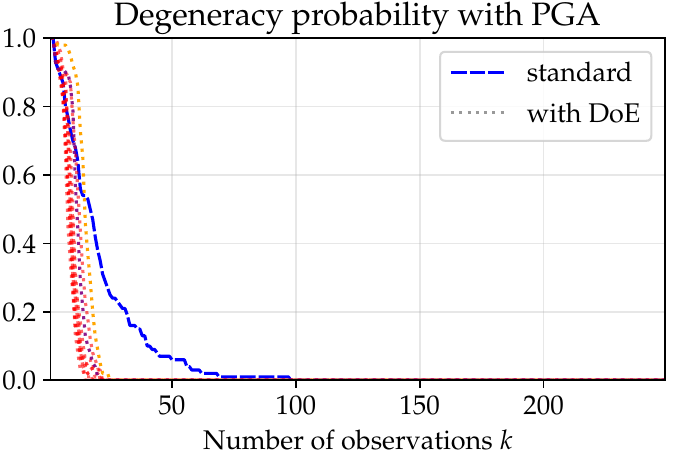}\\[0.7em]%
    \begin{minipage}{0.94\textwidth}\itshape\footnotesize 
        Left: the considered IM is the PSA. Right: it is the PGA.
    The degeneracy probability without DoE (blue dashed line) is compared with the ones with DoE (dotted lines), for different values of $\gamma$. 
    Two extreme values of $\gamma$ are highlighted: $\gamma=0.1$ (purple dotted line) and $\gamma=1.9$ (orange dotted line).
    \end{minipage}
    \caption{Probability that a sample of size $k$ yields a degenerate likelihood, as a function of $k$. }
    \label{fig:errors-degen}
    \end{figure*}

    \paragraph{Degeneracy disappearance} 

    {Definition~\ref{def:degeneracy} lists the different types of situations for which likelihood degeneracy can occur. For our case study, Fig.~\ref{fig:errors-degen} presents the average number of occurrences of the three types of situations that lead to a degenerate likelihood, as a function of the number of observed samples. We observe that the DoE method clearly outperforms in reducing the degeneracy for small numbers of observations compared with the standard method (from $k\simeq 20$ in the case of the PGA, and at most $k\simeq 50$ in the case of the PSA). Echoing this observation, we notice that this probability is, in general, higher with the PSA as IM, especially when the seismic signals are selected according to the standard method.

{To interpret these results, it is necessary to understand the fundamental difference between the PSA and the PGA. On the one hand, since the PSA is the indicator that is the most correlated with the structure’s response, the value of $\beta$ of the associated fragility curve is lower than with the PGA. With an ideal indicator ---which does not exist in practice unless the structure's response is known, which is meaningless--- this value would be equal to 0. On the other hand, the PSA distribution has a larger standard deviation than that of the PGA.

Thus, in the standard case, whether with the PSA or the PGA, as we have considered a high failure threshold, the main cause of degeneracy comes from the fact that the observed data are of type 1, that is to say they do not reveal any failure. With the PSA, the value of $\beta$ of the fragility curve being lower than with the PGA, while its distribution has a larger standard deviation, this favors a type 3 degeneracy. This explains a higher probability of degeneracy with the PSA than with the PGA.

To a lesser extent, this also affects the DoE method, although it has been shown in \ref{app:toycases} that, in the ideal case where the fragility curve is lognormal and we have data that exactly match the IM values proposed by the DoE algorithm, all cases are equivalent from the point of view of estimating $\theta$. That is, in an ideal case, we should observe no difference from the point of view of the efficiency of the DoE algorithm, whether the IM is the PGA and the PSA, or another IM.}

    \paragraph{Performances with few observations} {Focusing first on the quadratic errors presented in Figures~\ref{fig:errors-psa} and \ref{fig:errors-pga}, we show that when the number of observations $k$ is small, the outperformance of the DoE method over the standard one is clearly visible for any IM considered.}
    Indeed, on the domain $k<50$, the decrease of any of the errors is faster, reaching satisfying values at $k=50$, which is consistent with the qualitative study of the previous section.
    Regarding the cases where $k>50$, the behavior of the metrics differ between the case where the PSA is used and the one where it is the PGA. Globally, it is obvious that the PSA provides the best results. The consideration of this IM rather than the PGA is more relevant for this case study. {In the following paragraphs, we discuss more specifically the estimation biases for both IMs, since the quadratic error is a combination of the bias and the credibility width.}

    \paragraph{Study of the bias when the IM is the PSA}
    As mentioned in Section~\ref{sec:reference}, the PSA is the best of the two IMs in our case study. As a result, the performances of the DoE method is beyond doubt (Fig.~\ref{fig:errors-psa}). Up to values of $k$ close to $50$, the decrease in the estimation bias is rapid towards the so-called model bias value which is materialized in the figure (see Section~\ref{sec:metrics}). This bias reflects the fact that the reference fragility curve, $P_f^{\text{ref}}$, does not correspond to an exact probit-lognormal curve as illustrated in Fig.~\ref{fig:reference-frags}. Beyond a value of $k$ close to $100$, the bias stops decreasing significantly because it tends towards the model bias. From this threshold, additional observed data only make the credibility intervals thinner around the median. It is noted in this regard that the DoE method is able to provide estimates with a bias that is smaller than the so-called model bias. On average, the DoE-based estimations do not exactly match the model bias, although it is very close. This is not the case for the standard method, for which there is almost an order of magnitude difference, even with a sample size of $250$. The difference between the bias obtained on average by the DoE method and the model bias simply reflects the fact that the distributions of the data used for the estimates are not the same in the two cases. $P_f^{\text{MLE}}$ is estimated on the entire available database while the DoE methodology aims to select some of them by maximizing the criterion defined in Eq.~(\ref{eq:index}). \ref{app:toycases} provides more insight into the DoE approach.

\paragraph{Study of the bias when the IM is the PGA}
{As mentioned in section~\ref{sec:reference}, with the PGA as IM, it is not possible to fully describe the fragility curve. We therefore have no information on its evolution beyond the maximum value of the PGA observed, at about $12$ m/s$^2$. Moreover, the discrepancies between $P_f^{\text{ref}}$ and $P_f^{\text{MLE}}$ are maximal for a PGA value of the order of $8$ m/s$^2$, whereas before this value the fit is very good.} For the strongest seismic signals at disposal, equipment failure is only observed in around 70\% of cases.

{As shown in Fig.~\ref{fig:errors-pga}}, up to values of $k$ close to $100$, the DoE method outperforms the standard method. Beyond that, the behavior is different from that observed when the IM is the PSA and an evaluation of the performances of the DoE method is less obvious. On average, however, with the first $40$ data points, the bias value obtained with the PGA is smaller than with the PSA.

{In \ref{app:toycases} it is shown that with the DoE method, the distribution of the selected data tends to follow
a bimodal distribution whose modes are equally distant from the logarithm of the median of the fragility curve
(see also Fig.~\ref{fig:ex-estfrag} and the repartition of the ``red crosses''). Since knowledge of the fragility curve over a restricted domain is likely to reduce the performance of the DoE method, a thorough study of the consequences of such a restriction is proposed in \ref{app:toycases}, in a case where the model bias does not exist. This study shows that the performance of the DoE method is weakly affected by this limitation. We can therefore postulate that it is the presence of a significant bias in one of the two main domains in which the DoE method selects data that is the cause of the deterioration in the method's performance for values of $k > 100$, compared to the standard approach. If we believe the results in Fig.~\ref{fig:ex-estfrag}, with $80$ data, we obtain however a completely reasonable estimate, taking into account the observed bias.}

\begin{figure*}
        \centering%
        \hspace*{-1cm}%
        \includegraphics[width=5cm]{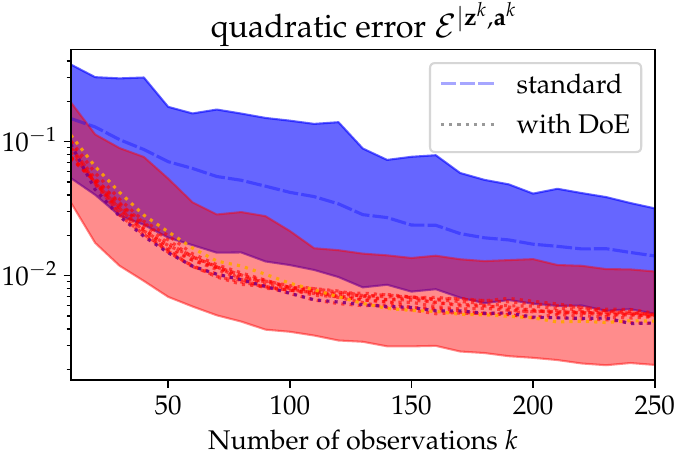}\includegraphics[width=5cm]{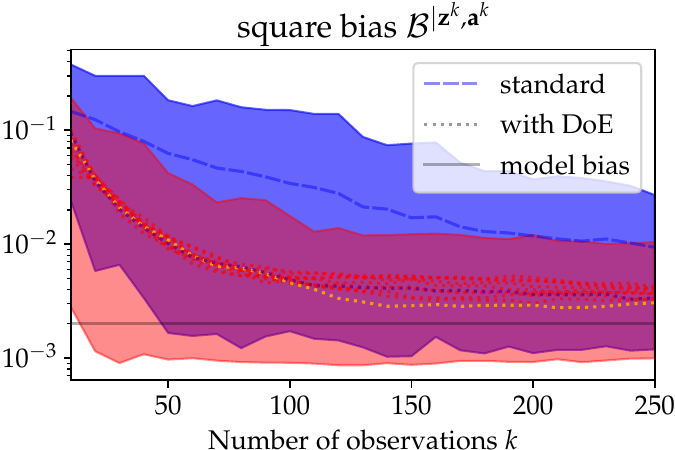}\includegraphics[width=5cm]{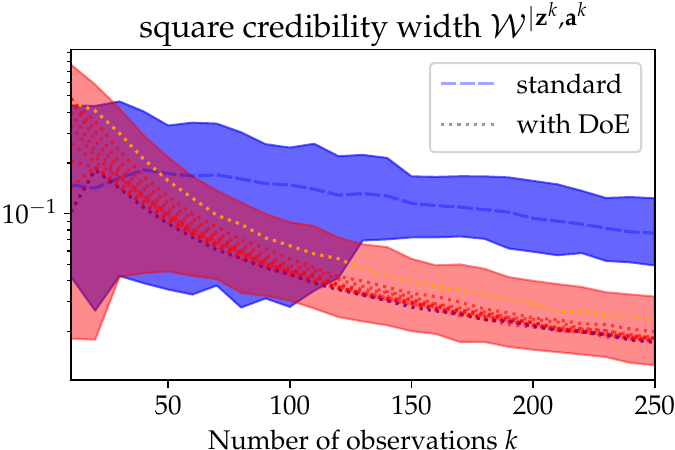}\\[0.7em]%
        \begin{minipage}{0.94\textwidth}\footnotesize\itshape 
            In dashed blue lines: the average errors using the standard distribution of the IM surrounded by their $95\%$-confidence interval. In dotted lines: the average errors using our DoE strategy with different values of $\gamma$. They are surrounded (in red) by the maximal $95\%$-confidence intervals. Two extreme values are emphasized: $\gamma=0$ (orange) and $\gamma=1.9$ (purple).
        On the middle figure, the model square bias is plotted as a gray line.
         The IM  is the PSA here.  
        \end{minipage}
        \caption{Average error metrics derived from numerous replications of the method, as a function of the size of the observed sample.}
        %
        %
        \label{fig:errors-psa}
    \end{figure*}

    \begin{figure*}
        \centering
        \hspace*{-1cm}%
        \includegraphics[width=5cm]{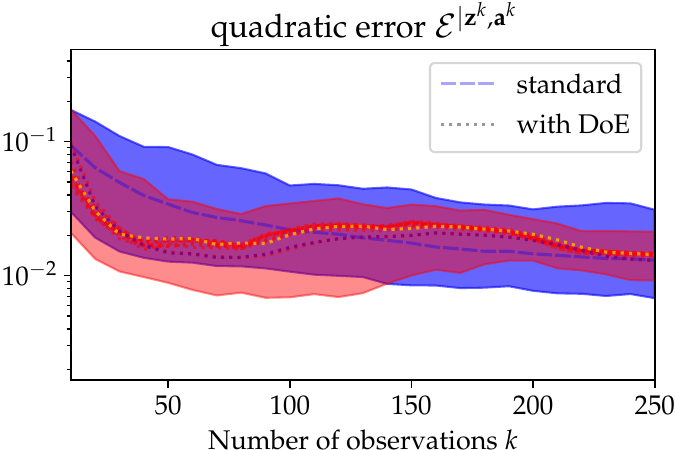}\includegraphics[width=5cm]{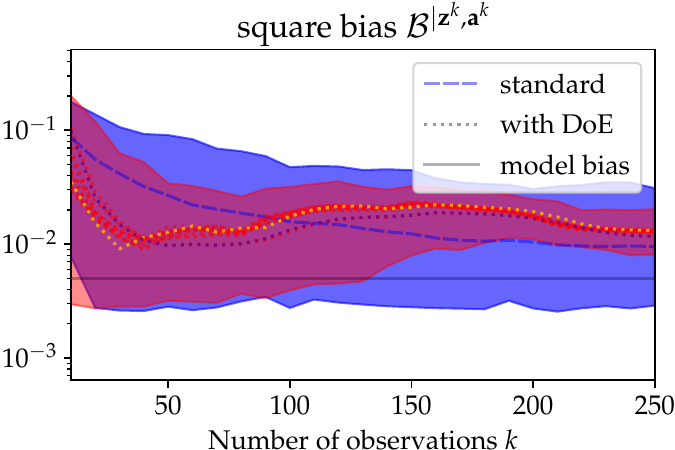}\includegraphics[width=5cm]{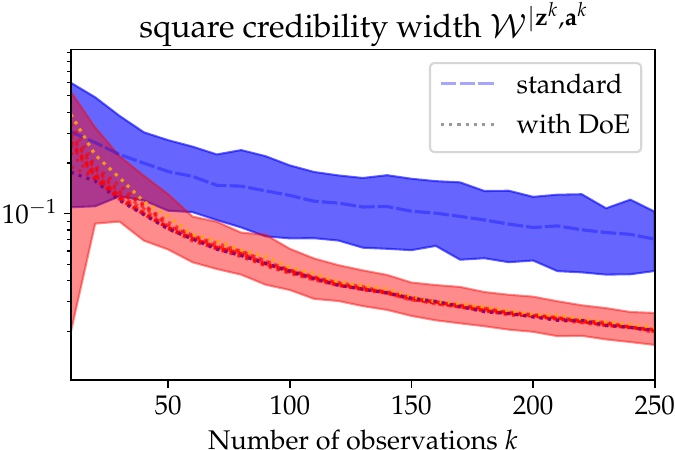}%
        \caption{As in Figure~\ref{fig:errors-psa} but here the IM  is the PGA.}
        \label{fig:errors-pga}
    \end{figure*}

\paragraph{{Stopping criterion}} {To the extent that an irreducible estimation bias can be expected in practice, an indication that such a bias has been reached would constitute a suitable stopping criterion for the DoE method. Since such a criterion would imply knowing, in some way, the bias itself and this is not possible with small data sizes, we suggest to study (i) the index $\cV\cI_k$ that measures the variation of the  quantity ~(\ref{eq:index}) that is used to select the seismic signals in the DoE method and (ii) the index $\cV\cP_k$ that measures the average evolution of the median of the fragility curve. These two quantities are defined in Section~\ref{sec:stopping_crit}. They reflect the information provided by adding a $k$-th data point. 

Fig.~\ref{fig:variaI} shows the evolutions of the average values of $\cV\cI_k$ for respectively the PSA and the PGA as IM, whereas Fig.~\ref{fig:variaP} is devoted to the evolutions of $\cV\cP_k$. These two figures, once again, show that learning with the DoE method is more effective than with the standard method. $\cV\cI_k$ and $\cV\cP_k$ are slightly influenced by the change in IM. Up to values of $k$ close to $k=50$, they even indicate that with PGA, the DoE method reaches more quickly a state in which the information brought by a new data point has less impact. This result is consistent with the one mentioned in the previous paragraph, namely that with the first $40$ data points, the bias value obtained with the PGA is smaller than with the PSA, with the DoE method.

The practitioner can therefore use these indicators to define a stopping criterion depending on a threshold value. For instance, in our case, a threshold value set at $\cV\cI_k = 10^{-3}$ gives approximately $k_{\max} = 60$ for the PSA and $k_{\max} = 40$ for the PGA. Concerning $\cV\cP_k$, the threshold value can be set, {for instance}, at $5 \% $: beyond this point, the index indicates that the estimated fragility curve exhibits minimal evolution when additional data are generated}. For both the PSA and the PGA, this also guarentees that the likelihood degeneracy probability is zero (see Fig.~\ref{fig:errors-degen}). {
These thresholds reflect our interpretation of the present case study and are not supported by existing guidance in the literature; practitioners should therefore select appropriate thresholds based on their objectives and the specific characteristics of the system they study}.
Remark that, even though the number of data required to reach a given stopping criterion is smaller with the PGA than with the PSA, this does not imply that the estimate is less accurate. Indeed, for small data sizes, the degeneracy phenomenon is more likely with the PSA than with the PGA (see Fig.~\ref{fig:errors-degen}). Consequently, the widths of the credibility intervals remain large for large $k$ with the PSA (see Figures~\ref{fig:errors-psa} and \ref{fig:errors-pga}). Our approach addresses this phenomenon by reducing the degeneracy even for a small number of observations. In practice, however, this number remains dependent on the IM of interest (see the paragraph ``Degeneracy disappearance'').

}

\begin{figure*}
    \centering%
    \includegraphics[width=5cm]{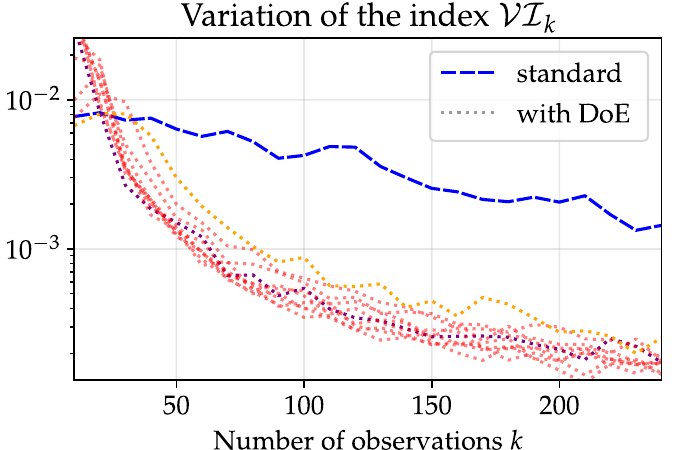}\ 
    \includegraphics[width=5cm]{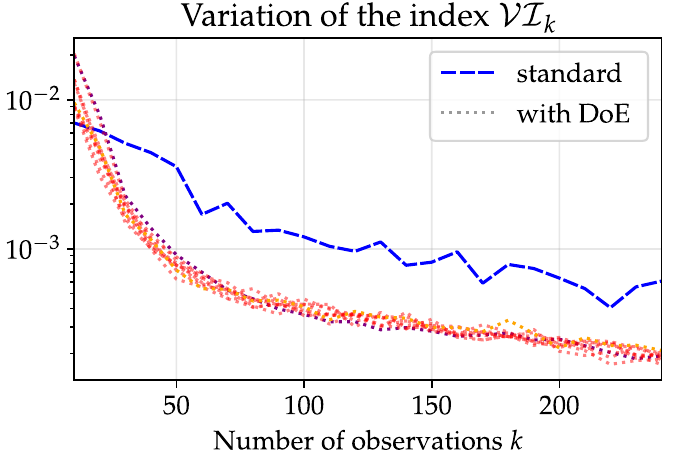}\\[0.7em]
    {\itshape\footnotesize Left: the considered IM is the PSA. Right: the considered IM is the PGA.}
    \caption{Average variations of the DoE index $\cV\cI_k$ as a function of the observed sample size $k$} %
    \label{fig:variaI}
\end{figure*}

\begin{figure*}
    \centering%
    \includegraphics[width=5cm]{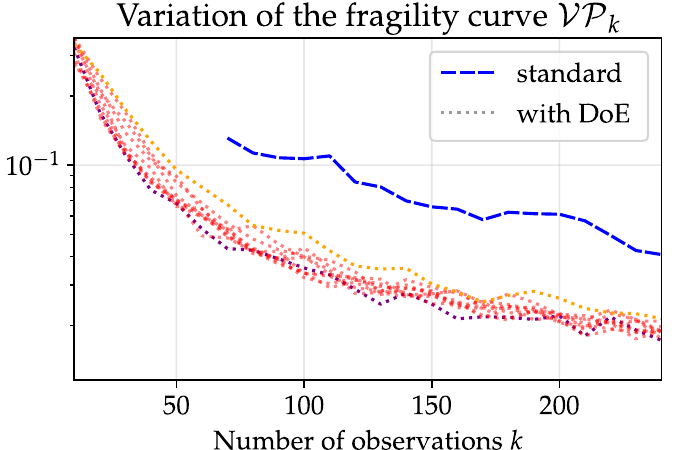}\ 
    \includegraphics[width=5cm]{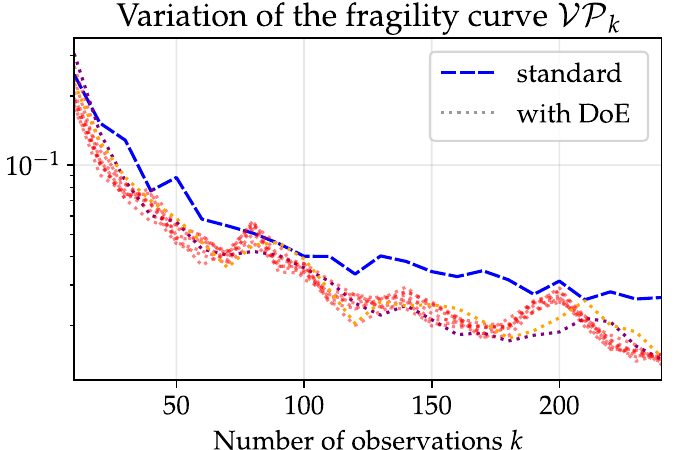}\\[0.7em]
    {\itshape\footnotesize Left: the considered IM is the PSA. Right: the considered IM is the PGA.}
    \caption{Average variations of the median \emph{a posteriori} fragility curve $\cV\cP_k$ as a function of the observed sample size $k$.} %
    \label{fig:variaP}
\end{figure*}

{\subsection{Synthesis} 

{Judging by its abundant use in the literature, the probit-lognormal model is nowadays considered to the practitioner as a model that is both pragmatic and relevant for estimating fragility curves. Its use with a limited amount of data is nevertheless a crucial point. Besides the facts that data is generally expensive ---especially for experimental tests--- and that the probit-lognormal model is likely to be biased in practice, there is no point in feeding it with a large amount of data. The objective of the methodology proposed in this work is then to make the most of the probit-lognormal model by obtaining the most accurate estimate possible with the minimum amount of data. To this end, we proposed a DoE methodology in a Bayesian framework based on the reference prior theory.

In practice, the choice of the best IM to consider for expressing fragility curves is an open subject. Although some criteria exist in the literature \cite{Cornell2004,Luco2007,Padgett2008}, there is no unique solution, as the choice ideally depends on the structure being studied. In practice, this choice may therefore be dictated by practical considerations. For example, at the scale of a large industrial facility, such as a nuclear power plant, the PGA may be chosen as IM for all relevant equipment to facilitate risk assessment studies, even if it is not optimal.

Although in \ref{app:toycases} it is demonstrated, through a toy case study, that in theory the performance of the DoE method is not affected by the choice of the IM, this is not the case in practice. What primarily affects its performance is the bias of the model, which can differ from one IM to another. So, in this work, two current IMs ---the PSA and the PGA--- were considered to assess the performance of the method which was applied to an equipment from the nuclear industry. Thus, it is shown that overall performance is best if the method is used with an IM that has a good level of correlation with the structural response of interest, such as the PSA in our setting. Nevertheless, in both cases, with a small number of data ---of the order of 80---, the performance of the proposed method is much better than that of the standard one. In our study, when the PGA is chosen as the IM, the method's performance degrades and is worse than with the standard method beyond 100 data points. This performance degradation is explained by the fact that the bias is more pronounced with the PGA in one of the two main domains where the DoE method selects data. 

These results reinforce the idea of using the probit-lognormal model only with small datasets, due to the inherent model bias that can be expected in practice. In order to guide the user, we also proposed two indicators that allow learning to be stopped based on quantitative information. Both criteria appeared sensitive to the change of IM and therefore to the associated potential biases. Let us add, however, that since the degeneracy phenomenon does not make the estimation optimal from the point of view of the size of the credibility interval, it is first of all appropriate, in practice, to continue learning until a non-degenerate sample is obtained (see Definition~\ref{def:degeneracy}).
}
}

\section{Conclusion \label{sec:conclusion}}

{Assessing the seismic fragility of structures and components when few binary data are available ---i.e., less than 100--- is a challenging task. To do this, it is often necessary to introduce some simplifying assumptions. A standard solution is to use a parametric model of the fragility curves such as the probit-lognormal model. This is the choice that was made in this work.

Several methods can be implemented to estimate the model parameters. The most recommended is the MLE-based method, coupled with a bootstrap approach to obtain confidence intervals. With very small numbers of data, the authors have however shown that the estimation could be affected by likelihood degeneracy problems which lead to inappropriate estimates of the fragility curves. To get around this problem, it is possible to resort to the use of Bayesian methods which are known to be regularizing. The question of the choice of the prior is then significant because it inevitably impacts all the resulting estimates, especially with small data sizes.

In this work, we have favored the Bayesian framework by relying on the reference prior theory in order to define the prior. This framework makes it possible to define a so-called objective prior in the sense that it only depends on the statistical model, namely the probit-lognormal model fed by binary data indicating the state of the structure of interest ---i.e., failure or non-failure--- for given IM values. The authors showed in a previous work that the solution to this problem is the well known Jeffreys prior. Since it is improper and does not allow to regularize the problem of degeneracy of the likelihood, a proper approximation of it has been proposed in this work. This results in a slightly informed prior that also has the advantage of being completely analytical, which facilitates its numerical implementation. {The proposed analytical expression is obtained by assuming a lognormal distribution for the IM. This assumption is consistent with the real seismic signals considered in this study, which correspond to a specific seismic scenario, but the literature shows that it is suitable for more general applications (near-source ground motions, artificial seismic signals generated from GMPE, etc.). Thus, thanks to its simple analytical formulation and its wide field of application, this prior can prove very useful in engineering practice using the Bayesian framework.} From this prior, authors have proposed a strategy of design of experiments inherited from the reference prior theory, i.e., based on the information theory. Given a large database of synthetic seismic signals, the proposed strategy intends to sequentially select the synthetic signals with which to perform the calculations or the tests, in order to optimally estimate ---by minimizing their number--- the probit-lognormal estimations of fragility curves.

Compared to a standard approach that aims to select seismic signals in their initial distribution, considering the Jeffreys prior, we have shown the superiority of our method. It allows to reach more quickly ---i.e., with a limited number of data--- a small estimation bias with a small credibility interval, while reducing the phenomenon of degeneracy. The proposed methodology therefore makes the most of the log-normal model. Since estimation biases can be expected in practice, it is recommended to use the method with few data (i.e., less than 100). To this end, we propose two stopping criteria that reflect the information provided by any additional data.

Finally, let us note that the DoE methodology was applied here within the framework of the reference prior theory because the authors believe that the objectivity of the prior is essential for fragility analyses. However, the DoE methodology can be applied with any prior that is proper.

}

\section*{Acknowledgement}

This research was supported by the CEA (French Alternative Energies and Atomic Energy Commission) and the SEISM Institute (\url{www.institut-seism.fr/en/})

\appendix

\section{Behavior of the DoE methodology on toy case studies}\label{app:toycases}

\subsection{Description of the case studies}

In order to gain more insight into the DoE methodology, we suggest in this appendix a brief analysis of its behavior on case studies that are free of any bias. 
Such case studies are purely theoretical, so that we qualify them as ``toy case studies''. Here, the observed data $(\mbf z^k,\mbf a^k)$ are perfectly generated according to the probit-lognormal statistical model presented in Section \ref{sec:model}:
a reference parameter $\theta^\ast=(\alpha^\ast,\beta^\ast)$ is chosen; then, from any input $a$ (that simulates a theoretical IM), an output $z\in\{0,1\}$ is generated that is equal to $1$ with probability $\Phi\left(\beta^{\ast -1}\log\frac{a}{\alpha^{\ast}}\right)$ and $0$ otherwise.

The selected inputs belong to an interval $\cA=(0,a_{\text{max}}]$.
They are picked by the design of experiments methodology that we present in this work with $\gamma=0.5$. Note that when implementing Algorithm \ref{alg:PE} in this section, there is no need to generate seismic signals from a particular value of $a$ here. Indeed the ``experiment'' (understand here: the generation of $z$) directly results from $a$.
Regarding the $k_0=2$ required initial values of $a$, we draw them uniformly in $\cA$.

In those toy case studies, the value of $\theta^\ast$ does not matter, in the sense that two such case studies with two different value of $\theta^\ast$ would be equivalent up to a translation and a dilatation w.r.t. $\log a$ (see \ref{app:equiv-toy} for a proof of this statement).

In this section, we present three different toy case studies. They are all implemented using the same value of $\theta^\ast$. 
They differ from the different domains $\cA$ that are selected for them.
The limits $a_{\text{max}}$ of these domains are chosen to simulate an upper bound 
on the possible IM that can result from a generated seismic signal.
They are selected to match a quantile of the reference fragility curve: if $q\in[0,1]$ one can derive $a_q$ such that $P^{\text{ref}}_f(a_q):=\Phi\left(\beta^{\ast -1}\log\frac{a_q}{\alpha^{\ast}}\right)=q$, i.e.
    \begin{equation*}
        a_q = \exp\left( \beta^\ast t_q+\log\alpha^\ast \right),
    \end{equation*}
where $t_q$ represents the $q$-quantile of a standard Gaussian distribution.

The three toy case studies implemented in this work result from defining $a_{\text{max}}=a_q$ with:
\begin{itemize}
    \item $q\simeq1$ (actually $q=1-10^{-3}$), to represent a case where no upper bound (or a very high one regarding the fragility curve) limits the generation of IMs. %
    \item $q=0.9$.
    \item $q=0.8$.
\end{itemize}
The reference fragility curve of each of these toy case studies are plotted on Fig.~\ref{fig:histograms_toy} (green solid lines) for each of the values of $q$ aforementioned, and for $\theta^\ast=(3,0.3)$. 

\subsection{Results}

The DoE strategy has been replicated a hundred times for each case study. Each time, we have stopped the algorithm after it reached a number of $250$ observed samples.
The last $100$ of these generated IMs have been kept for each replication. As we have carried out $100$ replications, they constitute a sample of $10^4$ values of IMs selected by the DoE algorithm.
The three empirical distributions of these samples resulting from the three values of $q$ are compared in Fig.~\ref{fig:histograms_toy}.

When the method is not limited by any upper bound for the choice of IMs (i.e. case $q=1$), 
the logarithms of the selected ones are symmetrically distributed around $\log\alpha^\ast$. This comment was expected as in the statistical modeling, $\log a$ is mathematically symmetric around $\log\alpha^\ast$ as well.

The bimodal appearance of the distribution supports the ``well-posedness'' of the method. Indeed, in a case where the result is deterministic, the knowledge of the evaluation of the curve in two points is enough to recover its parameter $\theta^\ast$. Thus, it makes sense that the optimal distribution of the IM for estimating the whole fragility curve is to retrieve a proper estimation of its evaluation in two zones of the domain.
A simple calculation is done in~\ref{sec:au-maths-calc} to suggest possible appropriate centers of the zones. 
{These are called $a^u_1$ and $a^u_2$ and are appropriate in the sense that they are the points that most help distinguish a stochastic fragility curve from a deterministic one.}
In Fig.~\ref{fig:histograms_toy}, their values are emphasized (dashed pink line). They are clearly close from the actual modes of the empirical distribution of the IM in the case $q=1$.

{In cases where the selection of values of $a$ is limited (i.e., when $q<1$), the method has been adapted to select $a = a_{\text{max}}$ when the algorithm suggests choosing a value larger than the maximum allowed value. We notice that the method seems to select the same ideal domain, in the broad sense, as when $q=1$ :  the data distribution maintains its bimodal character, the two modes being located in the vicinity of the two optimal domains. The method is robust: in both cases, $q_1 = 0.9$ and $q_2 = 0.8$, we observe similar proportions of $a$ values selected by the DoE algorithm, below the maximum IM values associated with each quantile, namely $a_{\text{max}}^1$ and $a_{\text{max}}^2$. Therefore, many values of $a$ are equal to $a_{\text{max}}^1$ or $a_{\text{max}}^2$, as appropriate.}

In Fig.~\ref{fig:metrics-toy}-left, we compare the average square bias
resulting from the 100 replications of the three case studies.
Their difference is difficult to distinguish, so the robustness of the method is confirmed:
even when the limit among the available IMs sharply intersects  the reference fragility curve, the distribution of the selected ones still provides an accurate estimation.
However, a close look allows to notice the expected loss of performance induced by such limit: the smaller is $q$, the larger is the bias.

Regarding the other metrics exposed in Fig.~\ref{fig:metrics-toy}, namely the indices $\cV\cI_k$ and $\cV\cP_k$, it is interesting to notice that 
their behaviors are comparable to that of the real case study.
In particular, $\cV\cI_k$ reaches the threshold $10^{-3}$ for a similar number of observations (around $50$) in any of the case studies implemented in this work (real ones as well as toy ones).
As a matter of fact, the index sequentially measures the quantity of information brought by the next seismic signal selected by the DoE. Thus, the stability of that quantity between case studies supports the robustness of our approach in selecting IMs to inform the estimate and enhance the learning.
However, the indices $\cV\cP_k$ are smaller for toy case studies. Without any model bias, the estimates reach more quickly the reference (see the average bias on Fig.~\ref{fig:metrics-toy}-left), so that they evolve less.

{Finally, this section helps to clarify the results obtained with the two IMs that are the PSA and the PGA in section \ref{sec:quantitative-study}. The main difference between them is that with the PGA it is not possible to completely describe the fragility curve. With the PGA, the maximum value of its probit-lognormal estimate is approximately $0.8$ but does not reach $1$ as with the PSA. If we ignore the bias, this means that the case $q=1$ is similar to the situation where the PSA is used as IM, while the cases  $q=0.8$ and $q=0.9$ are similar to the situation where the PGA is used. All things being equal, this study shows that the performance of the DoE method is weakly affected by a restricted learning domain. It is then likely that the presence of a significant bias in one of the two main domains in which the DoE method selects data is the cause of its performance degradation.}

\begin{figure*}
        \centering%
        \makebox[0pt][c]{%
        \includegraphics[width=5cm]{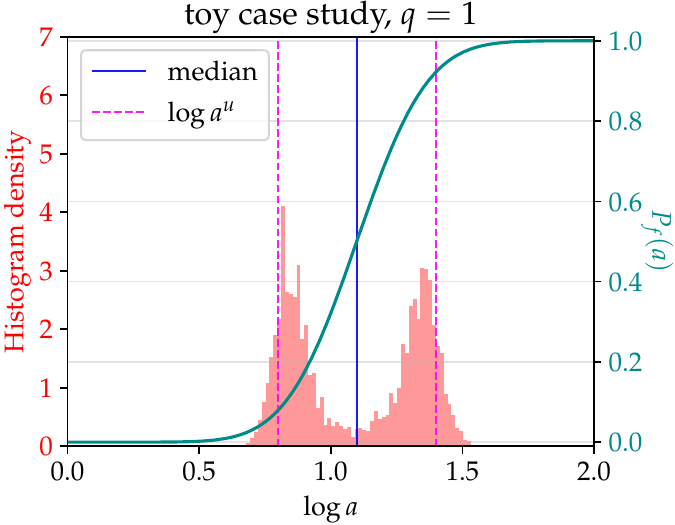}\ 
        \includegraphics[width=5cm]{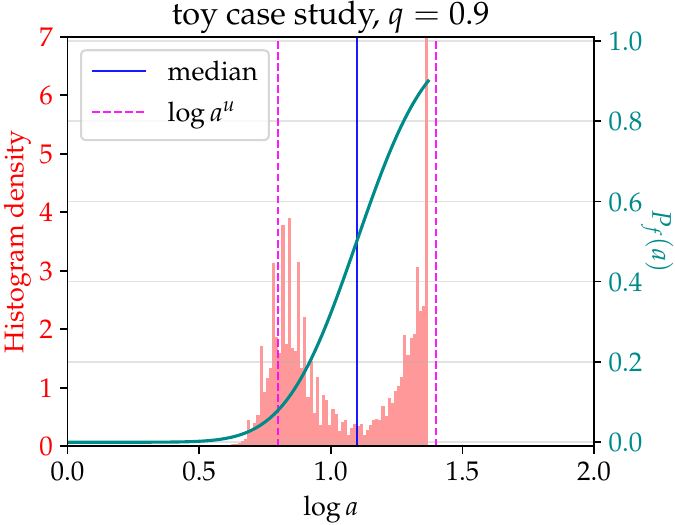}\ 
        \includegraphics[width=5cm]{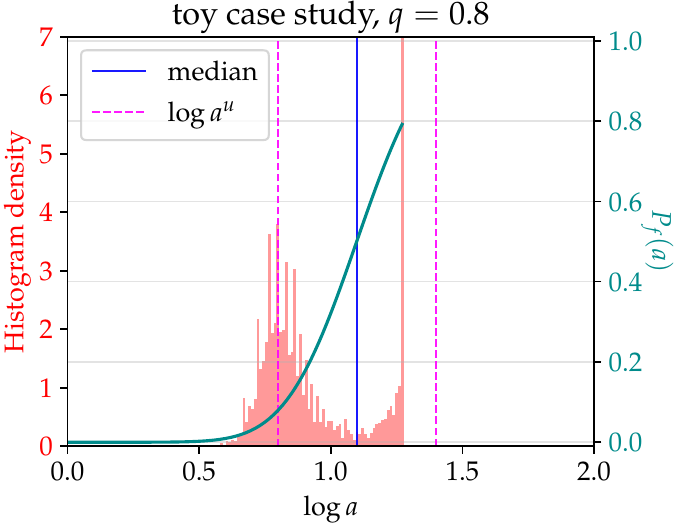}}\\[0.7em]%
            \begin{minipage}{0.94\textwidth}\footnotesize\itshape 
                For each of the three toy case studies: the distribution of the 100 last selected IM by the method (red histograms) along with the associated reference fragility curve (green) of the toy case study, which is cut at a certain quantile $q$ in each figure. The median of the reference (solid blue) and the values $a^u_{1}$, $a^u_{2}$ (dashed pink) are emphasized.
            \end{minipage}
        \caption{Stationary distributions of the selected IM by the method.}
        \label{fig:histograms_toy}
    \end{figure*}

\begin{figure*}
    \centering%
    \makebox[0pt][c]{%
    \includegraphics[width=5cm]{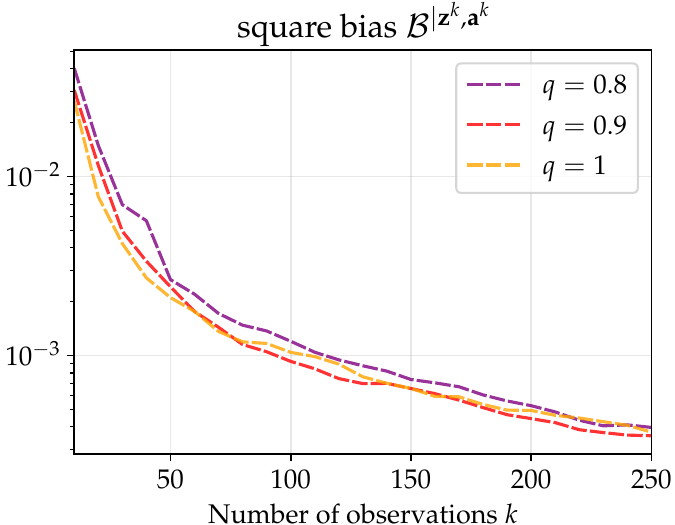}%
    \includegraphics[width=5cm]{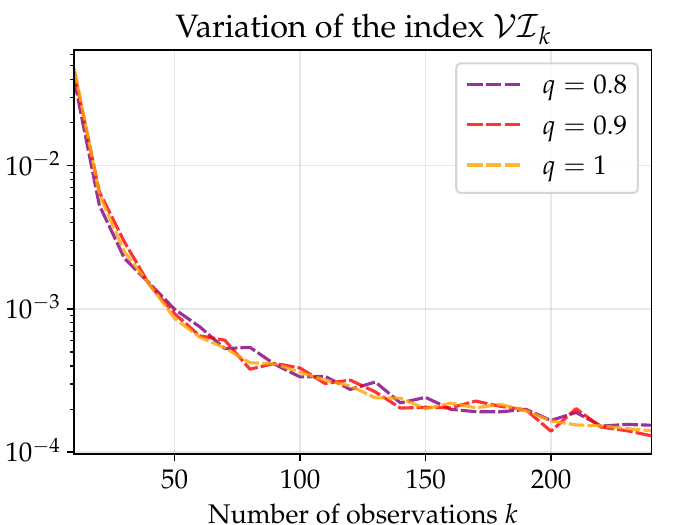}%
    \includegraphics[width=5cm]{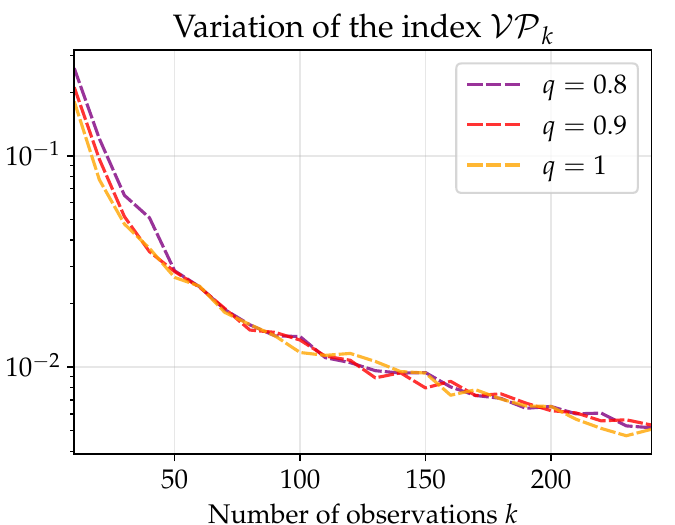}}%
    \caption{ Average values of $\cB^{|\mbf z^k,\mbf a^k}$, $\cV\cI_k$, $\cV\cP_k$ as functions of the number of observation $k$ for the three toy case studies.}
    \label{fig:metrics-toy}
\end{figure*}

\subsection{Simple suggestion of a domain where IMs should be theoretically selected}\label{sec:au-maths-calc}

In this section we suggest a criterion that gives insight on a domain where the
fragility curve should be evaluated in order to optimize the estimation.
{For this purpose, we study the quantity $\EE[ |P^{\text{ref}}_f(a)-P_f(a)| ]$ as a function of $a$, where $P^{\text{ref}}_f(a)$ represents a deterministic fragility curve and $P_f(a)$ denotes a stochastic estimate.
If $a^\ast$ is a maximal argument of this quantity, it would 
be the one that helps the most distinguishing the two curves.}

For this theoretical study, we assume to be in the settings of the toy case studies defined in this appendix, i.e. $P^{\text{ref}}_f(a)=\Phi\left(\beta^{\ast -1}\log\frac{a}{\alpha^{\ast}}\right)$ for a couple $(\alpha^\ast,\beta^\ast)$.
We derive
    \begin{multline}
        \hspace*{-0.5em}\frac{d}{da}\EE|P^{\text{ref}}_f(a)-P_f(a)|= \frac{1}{a\sqrt{2\pi}}\int_\Theta \left[\frac{1}{\beta^\ast}\exp\left(-\frac{(\log a-\log\alpha^\ast)^2}{2\beta^{\ast 2}}\right)   \right.\\
  \left.      -\frac{1}{\beta} \exp\left(-\frac{(\log a-\log\alpha)^2}{2\beta^2}\right)\right] p_{\alpha,\beta}(\alpha,\beta)d\alpha d\beta. \label{eq:derive-P-EP}   
    \end{multline}
To pursue, we can assume that $(\alpha,\beta)$ follows an inverse-gamma-normal distribution. This class covers a wide range of different distributions and is conjugate with a normal distribution, so that the following computation is tractable. Additionally, the probit-lognormal modeling of the fragility curve is equivalent to assume that there exists a latent variable $Y$ that is linearly correlated with the input: $Y=\log A +\cN(-\log\alpha^\ast,\beta^\ast)$. Considering this latent modeling, we can prove that both the reference prior and the posterior distribution of $\alpha,\beta$ belong to the class of inverse-gamma-normal ones \cite{VanBiesbroeckESAIMProcS}.
The inverse-gamma-lognormal distribution is defined by its density:
\begin{align*}
&    p_{\alpha,\beta}(\alpha,\beta) = K(c,d,\tau,\zeta)\left(\frac{1}{\beta^2}\right)^{c+1/2}\hspace*{-0.9em}\exp\left(-\frac{d}{\beta^2}\right)\exp\left(-\tau\frac{(\log\alpha-\zeta)^2}{2\beta^2}\right),\\
&    \text{with}\quad K(c,d,\tau,\zeta) =\frac{d^c\sqrt{\tau}}{\Gamma(c)\sqrt{2\pi}},
\end{align*}
where $c>0$, $d>0$, $\tau>0$ and $\zeta\in\RR$ are the parameters of the distribution. %

Thus, the right-hand term in  Eq.~(\ref{eq:derive-P-EP}) becomes:
\begin{align*}
    &\frac{1}{\sqrt{2\pi}}\frac{1}{a\beta^\ast}\exp\left(-\frac{(\log a-\log\alpha^\ast)^2}{2\beta^{\ast2}}\right)\\
    & - \frac{d^c}{(d+\frac{\tau}{2\tau+2}(\zeta-\log a)^2)^{c+1/2}}\frac{\sqrt{\tau}}{\sqrt{\tau+1}}\frac{\Gamma(c+\frac{1}{2})}{\Gamma(c)}\frac{1}{a\sqrt{2\pi}},
\end{align*}
so that it equals $0$ if and only if $a$ verifies:
\begin{align}\label{eq:equality-au}
   & |\zeta-\log a|\frac{e^{1/2}}{\beta^\ast}\exp\left(-\frac{(\log a-\log\alpha^\ast)^2}{2\beta^{\ast2}}\right)\\ &= \frac{d^c|\zeta-\log a|e^{1/2} }{(d+\frac{\tau}{2\tau+2}(\zeta-\log a)^2)^{c+1/2}}\frac{\sqrt{\tau}}{\sqrt{\tau+1}}\frac{\Gamma(c+\frac{1}{2})}{\Gamma(c)}.\nonumber
\end{align}
Given the derivations conducted in \ref{app:additionalmaths}, the right-hand term of the above equation tends to 1 when information is provided by observed data to the posterior. In this case and if $\zeta=\log\alpha^\ast$  the equation leads to $a=a^u_1$ or $a=a^u_2$ where $a^u_1$, $a^u_2$ are defined by
    \begin{equation*}
        a^u_{1,2} = \alpha^\ast\exp\left(\pm\beta^\ast\right).
    \end{equation*}

These values are built on a criterion that has a concrete and practical sense. However they remain purely theoretical as they depend on the exact value of $\theta^\ast$.

\subsection{Additional mathematical derivations}\label{app:additionalmaths}

    The goal of this section is to study the domain of the right-hand term of Eq. (\ref{eq:equality-au}), especially in the ``worst'' cases (the distribution is very-informative or non-informative).
    We rely on the following lemmas:
    \begin{lem}\label{lem:app-ineqG}
        For any $c>0$, $\displaystyle{\frac{\Gamma\left(c+\frac{1}{2}\right)}{\Gamma(c)}<\sqrt{c}}$. Also, $\displaystyle{
            \lim_{c\rightarrow\infty}\frac{\Gamma\left(c+\frac{1}{2}\right)}{\Gamma(c)\sqrt{c}}=1.
            }$
    \end{lem}
    \begin{proof}
        The first statement is a direct consequence of Gautschi's inequality \cite{Gautschi1959}: for any $c>0$, $s\in(0,1)$,
            \begin{equation*}
                c^{1-s}<\frac{\Gamma(c+1)}{\Gamma(c+s)}<(c+1)^{1-s}.
            \end{equation*}
        Fixing $s=1/2$ and using the identity $\Gamma(c+1)=c\Gamma(c)$ leads to the result.

        For the second statement, we can rely on Stirling's formula \cite{DavisGamma1959} that yields:
            \begin{equation*}
                \Gamma(c+t) \equi{c\rightarrow\infty} \Gamma(c)c^t
            \end{equation*}
        for any $t\in\CC$.
    \end{proof}
    
    \begin{lem}\label{lem:app-ineqdcf}
        Let $f>0$, for any $d,c>0$ we define:
        \begin{equation*}
            g(c,d) = \frac{\sqrt{c}}{(d+f)^{1/2}}\left(\frac{d}{d+f}\right)^c, \quad h(d) = \frac{1}{2}\log\left(1+\frac{f}{d}\right)^{-1}.
        \end{equation*}
        Thus, for any $c,d>0$, $g(c,d)<(2ef)^{-1/2}$, and $\displaystyle{\lim_{d\rightarrow\infty}g(h(d),d)=(2ef)^{-1/2}}.$
    \end{lem}
    \begin{proof}
        Let us differentiate $g$ w.r.t. $c$:
        \begin{equation*}
            \frac{\partial}{\partial c}g(c,d) = (d+f)^{-1/2}\left(\frac{d}{d+f}\right)^c\left[\sqrt{c}\log\left(\frac{d}{d+f}  \right)+\frac{1}{2\sqrt{c}}\right].
        \end{equation*}
        The above quantity is decreasing w.r.t. $c$ and equals $0$ when $c=h(d)$.
        We deduce that for any $d>0$, 
        \begin{equation}\label{eq:lem:ghdd}
            g(c,d)<g(h(d),d) = \frac{(2e)^{-1/2}}{(d+f)^{1/2}}\log\left(1+\frac{f}{d}\right)^{-1/2}.
        \end{equation}
        
        Now, let us briefly study the function $v(t)=(1+t^{-1})\log(1+t)$ for $t>0$.
        Firstly, $\log(1+t)\equi{t\rightarrow0}t$ so that $v(t)\conv{t\rightarrow0}1$. Secondly, $v'(t) = \frac{1}{t}-\frac{1}{t^2}\log(1+t)$, which is positive for any $t>0$, and tends to $\frac{1}{2}$ when $t=0$. Therefore, $v(t)>v(0)=1$ for any $t>0$. 

        Going back to Eq.~(\ref{eq:lem:ghdd}), we can write $g(h(d),d)=(2ef)^{-1/2}v(f/d)^{-1/2}$ to obtain that
            \begin{equation*}
                g(c,d) <(2ef)^{-1/2}\quad\text{and} \quad \lim_{d\rightarrow\infty}g(h(d),d)=(2ef)^{-1/2}.
            \end{equation*}
    \end{proof}

The result of Lemma \ref{lem:app-ineqdcf} with $f=\frac{1}{2}\frac{\tau}{\tau+1}(\log a-\zeta)^2$ lets us write the following inequality:
    \begin{equation*}
        \frac{d^c\sqrt{c}}{(d+\frac{\tau}{2\tau+2}(\zeta-\log a)^2 )^{c+1/2}}\frac{\sqrt{\tau}}{\sqrt{\tau +1}} < \frac{e^{-1/2}}{|\zeta-\log a|},
    \end{equation*}
    this upper-bound being reached at the boundary of the domain. Combining this statement with Lemma \ref{lem:app-ineqG}, we obtain: 
        \begin{equation*}
            \frac{d^c|\zeta-\log a|e^{1/2}}{(d+\frac{\tau}{2\tau+2}(\zeta-\log a)^2 )^{c+1/2}}\frac{\sqrt{\tau}}{\sqrt{\tau +1}}\frac{\Gamma(c+\frac{1}{2})}{\Gamma(c)} < 1,    
        \end{equation*}
    this upper-bound being reached at the boundary of the domain.

    Additionally, the left-hand term of the above equation is non-negative and tends to $0$ for some extreme values of the parameters.
    Therefore, there exist two extreme cases for the resolution of Eq.~(\ref{eq:equality-au}): when its right-hand term equals $0$ and when it equals $1$. The former corresponds to the limit case where the solution is $a=\pm\infty$ or $a=\zeta$, with $a=\pm\infty$ minimizing the quantity of interest. 
    The latter corresponds to the limit case where the solution verifies
        \begin{equation*}
            \frac{|\zeta-\log a|^2}{\beta^{\ast2}} = \exp\left( \frac{|\log\alpha^\ast-\log a|^2}{\beta^{\ast2}} -1\right).
        \end{equation*}
    The solutions of the above equation when $\zeta =\log\alpha^\ast$ are $a=\alpha^\ast\exp\left(\pm\beta^\ast\right).$
    They correspond to the limit case when the distribution $p_{\alpha,\beta}$ becomes informed and unbiased. 
    If the model is appropriately specified, that should be the fate of the posterior distribution: the posterior tends to match a normal distribution with mean $\theta^\ast$ and variance $\frac{1}{n}\cF^{-1}(\theta^\ast)$ \cite{VanDerVaart1992}.

\subsection{Equivalence between toy case studies}\label{app:equiv-toy}

    Let us consider two different toy case studies. They can differ by their reference parameters $\theta_1^\ast$, $\theta_2^\ast$ and their IMs $a_1$, $a_2$, which are defined in the domains $\cA_1$, $\cA_2$.
    We suppose that the bounds of $\cA_1$ and $\cA_2$ are defined to match the same quantiles of the reference fragility curves: $\cA_1=(c_1^1,c_1^2)$, $\cA_2=(c_2^1,c_2^1)$ with $P^{\mathrm{ref}}_1(c_1^1)=P^{\mathrm{ref}}_2(c_2^1)$ and $P^{\mathrm{ref}}_1(c_1^2)=P^{\mathrm{ref}}_2(c_2^2)$, where $P^{\mathrm{ref}}_1$ (reps. $P^{\mathrm{ref}}_2$) denotes the reference fragility curves given by $\theta_1^\ast$ (resp. $\theta^\ast_2$).

    Therefore, if we introduce $\tilde a$ that defines the value of a new IM on the second toy case study, which verifies
    \begin{equation*}
        \log\tilde a = \frac{\log c_1^2-\log c_1^1}{\log c_2^2-\log c_2^1}\log\frac{a_2}{c_2^1} + \log c_1^1,
    \end{equation*}
    we obtain that $\tilde a$ lives in the domain $\tilde\cA=\cA_1$. Given this new IM, the reference fragility curve of the second case study can be re-defined as a function of $\tilde a$:
        \begin{align*}
            &\tilde P^{\mathrm{ref}}_2(\tilde a)\\ &= \Phi\left(\beta^{\ast-1}_2\left[\frac{\log c_2^2-\log c_2^1}{\log c_1^2-\log c_1^1}\left(\log\tilde a-\log c_1^1\right)+\log c_1^1-\log\alpha^\ast_2 \right]\right) 
        \end{align*}
    for a certain $\tilde\theta^\ast_2 = (\tilde\alpha_2^\ast,\tilde\beta_2^\ast)$.
    We have $\tilde P^{\mathrm{ref}}_2(c_1^1)=P^{\mathrm{ref}}_2(c_2^1)=P^{\mathrm{ref}}_1(c_1^1)$ and $\tilde P^{\mathrm{ref}}_2(c_1^2)=P^{\mathrm{ref}}_2(c_2^2)=P^{\mathrm{ref}}_1(c_1^2)$, thus, as the parameter $\theta$ defines uniquely a probit-lognormal fragility curve, $\theta_1=\tilde\theta_2$.

    As a conclusion, given a rescaling of the IM, the two case studies are equivalent.

\end{document}